\documentclass[journal]{IEEEtran}
\usepackage{amsmath}
\usepackage{graphicx}
\usepackage[caption=false]{subfig}
\usepackage{float}
\usepackage{hyperref}
\newtheorem{theorem}{\bf Theorem} 
\ifCLASSINFOpdf
 
\else
 
\fi

\begin{document}

\title{Proactive Incentive Regulation in Multi-Agent Systems with Environmental Feedback}

\author{Xinyang Cao, Shijia Hua, and~Linjie Liu

\thanks{This work was funded by the National Natural Science Foundation of China (Nos. 62406255 and 62306243), China Postdoctoral Science Foundation (Certificate Number: 2024M762633), the Humanities and Social Sciences Research Planning Fund of the Ministry of Education (No. 24XJC630006), and Shaanxi Province Postdoctoral Research Fund (No. 2025BSHSDZZ134). (Corresponding author: Linjie Liu)}

\thanks{Xinyang Cao, Shijia Hua, and Linjie Liu are with the College of Science, Northwest A \& F University, Yangling, 712100.}} 

\markboth{Journal of \LaTeX\ Class Files,~Vol.~14, No.~8, August~2015}%
{Shell \MakeLowercase{\textit{et al.}}: Bare Demo of IEEEtran.cls for IEEE Journals}

\maketitle

\begin{abstract}
In environmental feedback systems, self-interested behaviors of rational agents often undermine cooperation and environmental sustainability. Although punitive incentive mechanisms are widely recognized as effective in addressing such social dilemmas, the timing of their implementation under different environmental conditions remains insufficiently understood. In this paper, we develop a multi-agent environmental feedback game framework with coupled incentive intensities across resource states. We investigate how system dynamics vary with incentive intensity under both resource-abundant and resource-scarce conditions. Theoretical analysis shows that imposing penalties in resource-abundant states can transform the system from a tragedy of the commons into a bistable regime that supports full cooperation. In contrast, imposing penalties in resource-scarce states leads to complex dynamical behaviors, including interior stability, heteroclinic cycles, and Hopf bifurcations, which hinder the emergence of stable cooperation. These results highlight the importance of timing in incentive regulation. Proactive intervention under favorable environmental conditions is more effective in sustaining cooperation than delayed enforcement after environmental degradation. 
\end{abstract}


\IEEEpeerreviewmaketitle

\section{Introduction}
In multi-agent systems, promoting the spontaneous emergence of cooperative behavior within a population remains a highly challenging issue \cite{perc2008social,xu2025reinforcement,perc2010coevolutionary,jia2025asymmetric}. The fundamental reason is that cooperation typically requires individuals to bear additional costs or sacrifice immediate local interests. Confronted with these costs, if every agent operates on absolute individual rationality, resorting to free-riding or over-exploiting resources to maximize personal gain, these micro-level selfish behaviors will inevitably precipitate a macro-level tragedy of the commons \cite{GDBJ1,jia2025social}, rendering the emergence of collective cooperation impossible. Meanwhile, the system environment is not a static backdrop. It provides continuous feedback in response to collective behaviors. When a substantial number of agents refuse to bear the costs of cooperation, the environmental state degrades accordingly \cite{han2024co}. This negative environmental feedback ultimately backfires on every individual within the system, resulting in the complete depletion of common-pool resources such as forests and fisheries \cite{gong2022limit,liu2023coevolutionary,GDBJBY}. This not only destroys the system's sustainability but also ultimately eradicates the agents' own long-term payoffs.

Evolutionary game theory offers a powerful framework for analyzing cooperation by modeling how strategy frequencies evolve based on interaction payoffs \cite{feng2023evolutionary,pi2025dynamic,zhu2025finite,han2022institutional,zhu2022nash,song2025emergence,wang2022modelling}. This evolution is typically represented by a payoff matrix and governed by replicator dynamics \cite{ramazi2020global}. The Prisoner’s Dilemma classically illustrates this challenge. Because defection consistently yields higher payoffs, rational individuals default to defection. Ultimately, defection dominates the population, directly mirroring the ``tragedy of the commons" \cite{axelrod1980effective,hardin1998extensions}. To sustain cooperation, researchers often introduce external interventions, such as incentive mechanisms, to alter the payoff matrix and discourage defection \cite{sun2023state,zhu2026evolutionary,hua2024coevolutionary,GEZHONGCUJINFANGSHI,JILI1,zhu2023equilibrium}. However, these conventional models assume a static payoff structure. In reality, individual behaviors actively alter the environment, which reciprocally reshapes the payoffs. Recognizing that payoff structures are not fixed, but dynamically co-evolve with the environmental state, has thus motivated the development of feedback-evolving games \cite{lu2024hybrid}.

The framework of feedback-evolving games captures a fundamental bidirectional loop. Individuals' behavioral decisions actively alter game environment, and these environmental shifts, in turn, drive changes in individual strategic choices. Early research into these coupled systems revealed that this dynamic interaction often results in oscillating tragedy of the commons \cite{Weiz}. Subsequent expansions of this framework have demonstrated that broader feedback conditions yield complex dynamic behaviors, including multiple steady states and chaos \cite{Tilman2018EvolutionaryGW}, which can be systematically classified using dilemma phase space analysis \cite{10.1093/pnasnexus/pgae455}. Within this evolving paradigm, research has increasingly explored the intersection of incentive mechanisms and dynamic environments \cite{JILI3,JILI4,CHENGFAGUODU,CHENGFA2}. Recent studies have investigated co-evolutionary dynamics across a spectrum of social dilemmas. Within these eco-evolutionary frameworks, the impact of external interventions has become a focal point \cite{hua2023facilitating,hua2024coevolutionary}. For instance, comparative analyses of reward and punishment mechanisms have demonstrated their distinct capacities to avert the tragedy of the commons.
Crucially, however, existing studies overlook a fundamental dimension: the environmental context of the interventions. It remains unclear whether applying incentives during periods of resource abundance or resource scarcity is more effective in fostering the emergence of cooperation.

In this paper, we develop a feedback-evolving game model of punitive incentives under both resource-rich and resource-scarce conditions in the Prisoner’s Dilemma. We consider that whenever an individual chooses to defect, the defector incurs a fine, applied separately in both resource-abundant and resource-scarce environments.
Through detailed theoretical analysis, we find that when incentives are applied under resource-abundant conditions, the system has the potential to converge to an ideal state of full cooperation with abundant resources. However, when incentives are applied under resource-scarce conditions, the system is difficult to achieve full cooperation. Although there is a possibility of escaping the tragedy of the commons, this is often accompanied by an extremely poor environmental state. Furthermore, applying incentives under resource scarcity can lead to complex dynamical phenomena such as heteroclinic cycle and Hopf bifurcations. These results indicate that punitive incentives are not effective in all scenarios. Their success is closely tied to the timing of their application. Applying penalties when resources are abundant helps ``prevent problems before they arise” and guides the system toward an ideal state. Conversely, applying penalties when resources are scarce often amounts to merely ``closing the barn door after the horse has bolted”. Our findings provide clear theoretical guidance for the design of penalties in the governance of public resources: penalties should be applied proactively when environmental conditions are still favorable, rather than waiting until resources are depleted to take remedial action. In summary, the main contributions of this work are as follows:
\begin{itemize}
\item We propose a state-dependent co-evolutionary framework that couples replicator dynamics with environmental transitions under punitive constraints.
\item We identify a diverse spectrum of dynamical regimes beyond fixed-point equilibria, specifically characterizing the emergence of expanding heteroclinic cycles and central periodic orbits within a Hamiltonian framework.
\item Our analysis demonstrates that proactive punishment in resource-abundant scenarios is superior to reactive intervention in resource-scarce ones.
\end{itemize}
The remainder of this paper is organized as follows. In Section II, we present our theoretical model and the dynamically feedback-evolving game system. Section III reports the main results, where we conduct a detailed theoretical analysis of the dynamical feedback system and provide concrete numerical examples to verify our analytical findings. Section IV offers a thorough discussion, in which we summarize our work and explore promising directions for future research.

\section{Model and Methods}
\begin{table*}[!t]
\centering
\caption{Stability of the equilibrium point when penalties are imposed under resource abundance}
\label{tab:combined_stability}
\begin{tabular}{c c c}
\hline
Equilibrium points & Existence condition & Stability condition \\
\hline
$(0,0)$ & Always exist & Locally stable \\
$(0,1)$ & Always exist & Unstable saddle when $0 < r < b$;
Unstable when $r > b$\\
$(1,0)$ & Always exist & Unstable saddle \\
$(1,1)$ & Always exist & Locally stable when $r > a$ \\
$(x_{n=1},1)$ & \begin{tabular}{c}
Exists when $a>b$ and $b<r<a$, \\
or $a<b$ and $a<r<b$
\end{tabular} & 
Locally stable when $r_{c1} < r < a$ (if $a>b$); \\
$(x^{*},n^{*})$ & Exists when $r > r_{c1}$ & Unstable saddle when exists \\
\hline
\end{tabular}
\label{table1}
\end{table*}
We consider an infinitely large, well-mixed population where individuals interact in pairwise games. Each individual adopts one of two strategies: cooperation ($C$) or defection ($D$). Environmental resource availability is captured by $n \in [0,1]$, which defines the probability of the environment being in an abundant state. Here, $n=1$ and $n=0$ correspond to absolute abundance and total depletion, respectively. Thus, the game structure under conditions of resource abundance and resource scarcity can be described by the following $2\times2$ payoff matrix:
\begin{equation}
    {A}_{0}= \begin{pmatrix}
{R}_{0}&{S}_{0}\\ 
{T}_{0}&{P}_{0}
\end{pmatrix}
\label{1}
\end{equation}
\begin{equation}
    {A}_{1}= \begin{pmatrix}
{R}_{1} & {S}_{1}\\ 
{T}_{1} &{ P}_{1}
\end{pmatrix}.
\label{2}
\end{equation}

Here, ${R}_{i}$ (reward) represents the payoff when both parties choose to cooperate; ${S}_{i}$  (sucker’s payoff) is the payoff for cooperating while the opponent defects; ${T}_{i}$ (temptation) is the payoff for defecting while the opponent cooperates; and ${P}_{i}$ (punishment) is the payoff for mutual defection \cite{10.1093/pnasnexus/pgae455}. These values must satisfy the structural characteristics of the Prisoner’s Dilemma game: ${T}_{i}> {R}_{i}> {P}_{i}> {S}_{i}$. As the environmental state dynamically changes, the payoff matrix for the actual game can be derived from the environmental resource abundance $n$ through linear weighting according to previous work \cite{Weiz,10.1093/pnasnexus/pgae455}:
\begin{equation}
    A_n = (1 - n)A_0 + nA_1.
    \label{3}
\end{equation}

To examine how the timing of incentive implementation affects cooperative behavior, we investigate two distinct incentive regimes: one in which penalties are imposed when environmental resources are abundant, and another in which penalties are activated only when resources become scarce. Let the penalty intensity be $r > 0$. The payoff matrix for imposing penalties when resources are abundant ($n=1$) is as follows:\cite{JILIYINRU}
\begin{equation}
    A(n) = (1 - n)\begin{bmatrix}
R_0 & S_0 \\
T_0 & P_0
\end{bmatrix} + n\begin{bmatrix}
R_1 & S_1 \\
T_1 - r & P_1 - r
\end{bmatrix}.
\label{4}
\end{equation}

The payoff matrix for imposing penalties during times of resource scarcity ($n=0$) is as follows:
\begin{equation}
    A(n) = (1 - n)\begin{bmatrix}
R_0 & S_0 \\
T_0 - r & P_0 - r
\end{bmatrix} + n\begin{bmatrix}
R_1 & S_1 \\
T_1 & P_1
\end{bmatrix}.
\label{5}
\end{equation}

To simplify the analysis, we define the following four parameters to characterize the relative attractiveness of defection or cooperation: $a= {T}_{1}- {R}_{1}$ represents the additional payoff of defection relative to cooperation when the opponent chooses to cooperate under resource abundance; $b= {P}_{1}- {S}_{1}$ represents the additional payoff of defection relative to cooperation when the opponent chooses to defect under resource abundance; $c = T_0 - R_0$ denotes the additional payoff of defection relative to cooperation when the opponent chooses to cooperate under resource scarcity; $d = P_0 - S_0$ denotes the additional payoff of defection relative to cooperation when the opponent chooses to defect under resource scarcity.

To characterize the dynamic changes in the level of cooperation, we derive the system’s replicator equations as follows \cite{FUZHIFANGCHENG}
\begin{equation}
    \dot{x} = x(1 - x)f(x,n),
    \label{6}
\end{equation}
where $x$ represents the proportion of cooperators in the population, and $f(x,n)$ denotes the difference between the fitness of a cooperating individual and that of a defector. When $f(x,n) > 0$, the frequency of cooperation increases. Conversely, when $f(x,n) < 0$, the frequency of defection increases. 

The evolution of the environmental resource abundance $n$ is influenced by group behavior. We assume that cooperation promotes resource recovery, while defection accelerates resource depletion. Thus, the following environmental dynamics equation holds \cite{Weiz}
\begin{equation}
    \dot{n} = \varepsilon n(1 - n)[(1 + \theta)x - 1],
    \label{7}
\end{equation}
where $\varepsilon  > 0$ is the environmental change rate coefficient, which controls the timescale of environmental evolution relative to strategy evolution. The larger $\varepsilon $ is, the faster the environment responds to behavioral changes; the smaller $\varepsilon $ is, the slower the environmental changes. $\theta $ is the ratio of the gain rate to the loss rate for cooperators and defectors. A larger $\theta$ implies that cooperators make a more significant positive contribution to resources relative to defectors, meaning fewer cooperators are needed to restore resources, whereas a smaller $\theta$ requires more cooperators.

In the payoff matrix, the signs of $T-R$ and $P-S$ determine the game structure of the system. Therefore, to characterize the dynamic changes in the game structure, we often introduce two dimensionless indices \cite{10.1093/pnasnexus/pgae455}
\begin{equation}
    D_g' = \frac{T - R}{R - P},
    \label{8}
\end{equation}
\begin{equation}
    D_r' = \frac{P - S}{R - P}.
    \label{9}
\end{equation}

For convenience, we also denote $D_g'$ as GID, which represents a player’s exploitation motivation to take advantage of the opponent \cite{ito2018scaling,tanimoto2015fundamentals,wang2015universal}. Similarly, we denote $D_r'$ as RAD, which reflects a player’s defensive motivation to avoid being exploited by others. Based on the signs of these two measures, the game structure can be classified into the following four types: When $D_g' > 0$ and $D_r' > 0$, the game type is a PD game, in which defection dominates; when $D_g' > 0$ and $D_r' < 0$, the game type is a Chicken game, where cooperation and defection coexist stably; when $D_g' < 0$ and $D_r' > 0$, the game is a SH game, and the system exhibits a bistable state; when $D_g' < 0$ and $D_r' < 0$, the game is a Trivial game, and cooperation dominates. The introduction and classification of this dimensionless index provide a powerful tool for subsequent comprehensive analysis of the system’s dynamical behavior \cite{10.1093/pnasnexus/pgae455}.
\section{Theoretical analysis and numerical examples}
\subsection{Evolutionary Dynamics Analysis of Punishment Under Resource Abundance}
In this section, we first analyze the coevolutionary dynamics of imposing punishment when resources are abundant. To begin with, we present the conditions for the existence of equilibrium points in the system.

\begin{theorem}\label{thm:main}
Apart from the corner equilibrium points $(0,0)$, $(0,1)$, $(1,0)$, and $(1,1)$, there exists an internal equilibrium point $(x^*,n^*)$ when $r > r_{c1} = \frac{\theta b + a}{1 + \theta}$. The existence conditions for the boundary equilibrium points of the system depend on the magnitudes of $a$ and $b$: 
\begin{itemize}
\item when $a > b$, the existence condition is $b < r < a$; 
\item when $a < b$, the existence condition is $a < r < b$; 
\item $(x_{n = 0}, 0)$ does not exist in this system. 
\end{itemize}
Where 
\begin{equation}
    (x^{\ast},n^{\ast}) = \left( \frac{1}{1 + \theta},\frac{\theta d + c}{\theta d + c - (\theta b + a) + (1 + \theta)r} \right)
    \label{10}
\end{equation}
\begin{equation}
(x_{n = 1},1) = \left( \frac{r - b}{a - b},1 \right)
\label{11}
\end{equation}
\begin{equation}
(x_{n = 0},0) = \left( \frac{d}{d - c},0 \right).    
\label{12}
\end{equation}
\end{theorem}

\begin{IEEEproof} First, for the equilibrium points other than the corner equilibrium points, we can set $\dot{x} = \dot{n} = 0$, $n = 1$, or $n = 0$ to determine the system’s internal equilibrium points $(x^*,n^*)$ and the boundary equilibrium points $(x_{n=1}, 1)$ and $(x_{n = 0}, 0)$, respectively.

Regarding the existence conditions for the internal equilibrium points $(x^*,n^*)$. It is clear that $0 < x^* < 1$. Assume that $0<\frac{\theta d + c}{\theta d + c - (\theta b + a) + (1 + \theta)r }<1$, so it must be that $\theta d + c - (\theta b + a) + (1 + \theta)r > \theta d + c > 0$. Thus, the existence condition can be expressed as
\begin{equation}
    r > r_{c1} = \frac{\theta b + a}{1 + \theta}.
    \label{13}
\end{equation}

For the boundary equilibrium points $(x_{n = 0}, 0)$, since $d > d - c$, the equilibrium points $(x_{n = 0}, 0)$ do not exist when the punishment is applied under the given conditions.

Regarding the existence conditions for the boundary equilibrium point $(x_{n=1}, 1)$. Since $0 < x_{n=1} < 1$, it must be the case that $a - b < r - b < 0$ or $a - b > r - b > 0$ . Solving the two equations separately yields the existence conditions for $(x_{n=1}, 1)$. When $a > b$, the solution is $b < r < a$; when $a < b$, the solution is $a < r < b$.
\end{IEEEproof}

In the following theorem, we present the stability conditions for the six equilibrium points:

\begin{theorem} The stability of the six equilibrium points of the system satisfies the following conditions:

2.1 The stability of the corner points $(0,0)$ and $(1,0)$ is independent of the penalty strength $r$. $(0,0)$ is a locally stable point, and $(1,0)$ is an unstable saddle point.

2.2 When $0 \leq r < a$, the system has a stable point $(0,0)$.

2.3 When $r > a$, the system has stable points $(0,0)$ and $(1,1)$.

2.4 In particular, assume that $a > b$ holds, the system has stable points $(0,0)$ and $(x_{n=1}, 1)$ when $r_{c1} < r < a$.

\end{theorem}
\begin{IEEEproof} The stability of these equilibrium points can be determined by the signs of the eigenvalues of the Jacobian matrix. At the point $(x, n)$, the Jacobian matrix takes the following form:
\begin{equation}
    J = \begin{pmatrix}
\frac{\partial \dot{x}}{\partial x} & \frac{\partial \dot{x}}{\partial n} \\
\frac{\partial \dot{n}}{\partial x} & \frac{\partial \dot{n}}{\partial n}
\end{pmatrix}.
\label{14}
\end{equation}

Specifically, we can derive the Jacobian matrix for the equilibrium point as follows:

For the equilibrium point $(0, 0)$, we have:
\begin{equation}
    J(0,0) = \begin{pmatrix}
-d & 0 \\
0 & -\varepsilon
\end{pmatrix}.
\label{15}
\end{equation}

For the equilibrium point $(0, 1)$, we have:
\begin{equation}
J(0,1) = \begin{pmatrix}
r - b & 0 \\
0 & \varepsilon
\end{pmatrix}.
\label{16}
\end{equation}

For the equilibrium point $(1, 0)$, we have:
\begin{equation}
J(1,0) = \begin{pmatrix}
c & 0 \\
0 & \varepsilon\theta
\end{pmatrix}.    
\label{17}
\end{equation}

For the equilibrium point $(1, 1)$, we have:
\begin{equation}
J(1,1) = \begin{pmatrix}
a - r & 0 \\
0 & -\varepsilon\theta
\end{pmatrix}    
\label{18}
\end{equation}

For the equilibrium point $(x_{n=1}, 1)$, we have:
\begin{equation}
J(x_{n=1}, 1) = 
\begin{pmatrix}
\frac{(a-r)(b-r)}{a-b} & x_{n=1}(1 - x_{n=1}) \frac{\partial f}{\partial n}(x_{n=1}) \\
0 & \varepsilon[1 - (1 + \theta)x_{n=1}]
\end{pmatrix} . 
\label{19}
\end{equation}

For the equilibrium point $(x^*,n^*)$, we have:
\begin{equation}
J(x^*, n^*) = \begin{pmatrix}
\frac{\theta}{(1 + \theta)^2} \frac{\partial f}{\partial x}(n^*) & \frac{\theta}{(1 + \theta)^2} \frac{\partial f}{\partial n}(x^*) \\
\varepsilon n^*(1 - n^*)(1 + \theta) & 0
\end{pmatrix}    .
\label{20}
\end{equation}

where :
\begin{equation}
\frac{\partial f}{\partial x} = (1 + \theta)[n(b - r - d) + d],
\label{21}
\end{equation}

\begin{equation}
\frac{\partial f}{\partial n} = \frac{1}{n}[d + x(c - d)].    
\label{22}
\end{equation}

The stability of the equilibrium points can be determined based on the signs of the determinant and trace of the Jacobian matrix. We have
\begin{itemize}
\item When Tr$(J(x,n)) < 0$ and Det$(J(x,n)) > 0$, $(x, n)$ is a locally stable point; 
\item When Tr$(J(x,n)) > 0$ and Det$(J(x,n)) > 0$, $(x, n)$ is an unstable focus; \item When Det$(J(x,n)) < 0$, $(x, n)$ is a saddle point; 
\item When Tr$(J(x,n)) = 0$ and Det$(J(x,n)) > 0$, $(x, n)$ is a center. 
\end{itemize}
We select the relatively complex equilibrium points $(x_{n = 1}, 1)$ and $(x^*, n^*)$ for analysis. For the equilibrium point $(x_{n = 1}, 1)$, the determinant and trace of its Jacobian matrix take the following forms
\begin{equation}
\text{Tr}(J(x_{n=1}, 1)) = \frac{(a - r)(b - r)}{a - b} + \varepsilon[1 - (1 + \theta)x_{n=1}],
\label{23}
\end{equation}

\begin{equation}
\text{Det}(J(x_{n=1}, 1)) = \frac{\varepsilon(a - r)(b - r)[1 - (1 + \theta)x_{n=1}]}{a - b}.
\label{24}
\end{equation}

Regarding its determinant, we can conclude that:
\begin{equation}
    \operatorname{sgn}\bigl(\text{Det}(J(x_{n=1}, 1)) \bigr) = \operatorname{sgn}\left( -\frac{1 - (1 + \theta) \frac{r - b}{a - b}}{a - b} \right).
    \label{25}
\end{equation}

We can treat $\frac{1  - (1 + \theta) \frac{r - b}{a - b}}{a - b} $ as a continuous linear function of $r$, and find that its zero is exactly $r_{c1}$. Combining this with the existence conditions for $(x_{n = 1}, 1)$, we can state:

When $a>b$, if $b<r<r_{c1}$, then $\text{Det}(J(x_{n = 1}, 1))<0$; when $r_{c1} < r < a$, $\text{Det}(J(x_{n = 1}, 1))>0$. When $a < b$, if $a < r < r_{c1}$, then $\text{Det}(J(x_{n = 1}, 1))<0$; if $r_{c1} < r < b$, then $\text{Det}(J(x_{n = 1}, 1))>0$.

The analysis of its trace is somewhat complex. Let us first simplify it as follows:

\begin{equation}
\begin{split}
    &\text{Tr}(J(x_{n=1},1)) = \\
    &\frac{r^2 - [a + b + \varepsilon(1 + \theta)]r + \varepsilon(\theta b + a) + ab}{a - b}
\end{split}
\label{26}.
\end{equation}

We can regard this as a continuous quadratic function of $r$. Therefore, we have: As $r \to a$, $\text{Tr}(J(x_{n=1},1)) \to -\theta\varepsilon < 0$; As $r \to b$, $\text{Tr}(J(x_{n=1},1)) \to \varepsilon > 0$. In particular, when $r = r_{c1}$, $\varepsilon[1 - (1 + \theta)x_{n=1}] = 0$. In this case, when $a > b$, $\text{Tr}(J(x_{n=1},1)) < 0$; When $a < b$, $\text{Tr}(J(x_{n=1},1)) > 0$. Due to its continuity, we know that there must exist a sign-changing point $r_1$ located either in $(b, r_{c1})$ or in $(a, r_{c1})$, which can be calculated using the discriminant of a quadratic equation. Combining this with the fact that $(x_{n=1},1)$ lies on the interval and has boundary values, we can analyze that:
\begin{figure}[!t]
    \centering
    \includegraphics[width=0.5\textwidth]{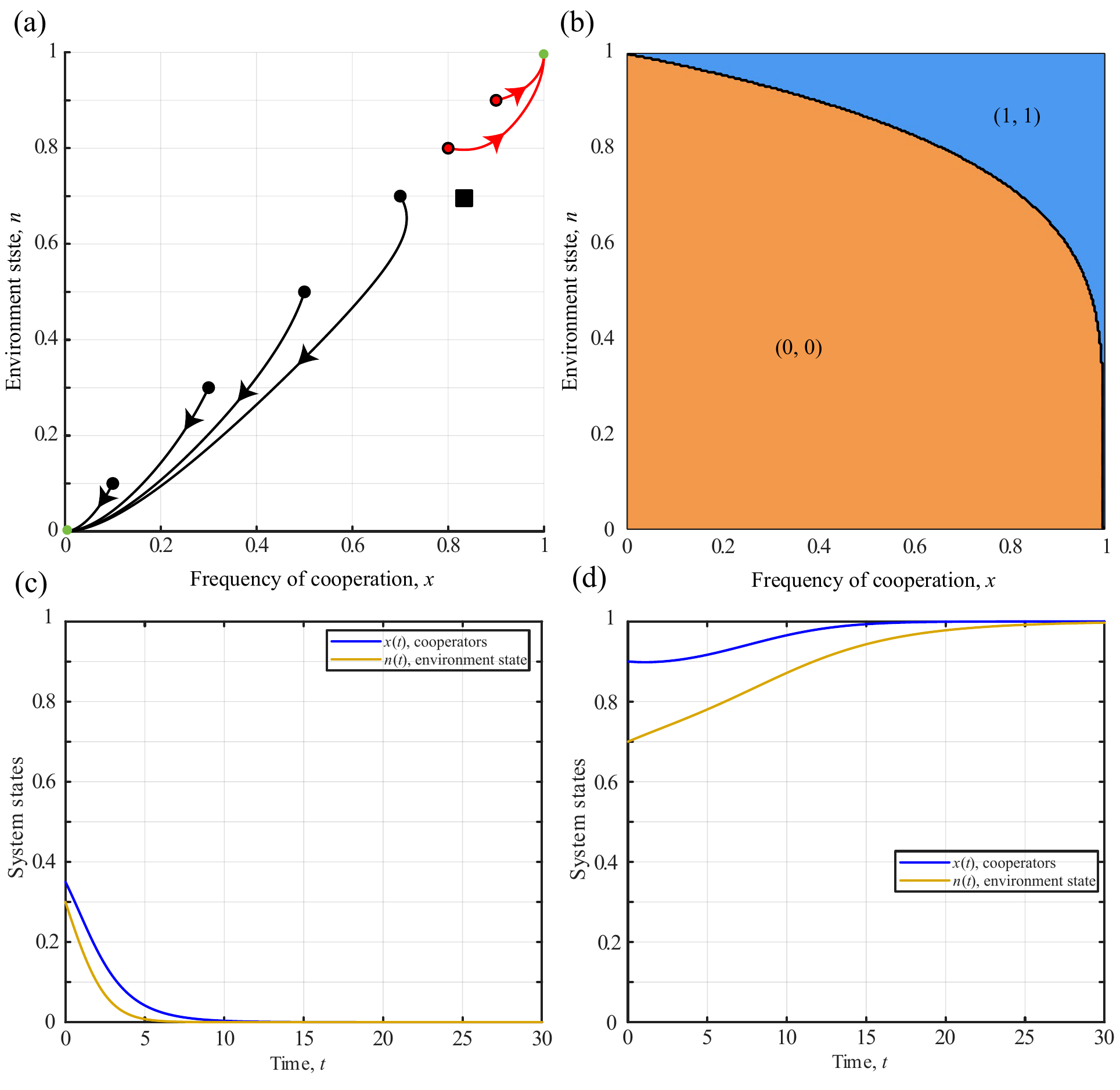}
\caption{Numerical verification of the system's bistability under clause 2.3 of Theorem 2. Panel (a) shows the phase portrait of the system in the $(x, n)$ plane, where black trajectories converge to $(0,0)$ and red trajectories converge to $(1,1)$. Panel (b) illustrates the attractor regions of these two stable points, with the orange and blue regions representing the attractor regions of $(0,0)$ and $(1,1)$, respectively. Panel (c) shows the time evolution of the system state with initial conditions $x=0.35$ and $n = 0.3$. Panel (d) shows the time evolution of the system state with initial conditions $x = 0.9$ and $n = 0.7$. Parameters are set as $x = 0.35$ and $n = 0.3$. Parameters are set as $R_0 = 3$, $S_0 = 2$, $T_0 = 4.5$, $P_0 = 2.5$, $R_1 = 3.5$, $S_1 = 2$, $T_1 = 5$, $P_1 = 3$, $\varepsilon = 1$, $\theta = 0.2$, $r = 2$.}
\label{Figure1}
\end{figure}

When $a > b$, 
\begin{itemize}
\item if $b < r < r_1$, then $\text{Tr}(J(x_{n=1},1)) > 0$;  
\item if $r_1 < r < a$, then $\text{Tr}(J(x_{n=1},1)) < 0$. 
\end{itemize}

When $a < b$, 
\begin{itemize}
\item if $a < r < r_1$, then $\text{Tr}(J(x_{n=1},1)) < 0$;  
\item if $r_1 < r < b$, then $\text{Tr}(J(x_{n=1},1)) > 0$. 
\end{itemize}
Combining the signs of the determinant and trace, we can derive the conclusion of Theorem 2.

For the equilibrium point $(x^*, n^*)$, the determinant and trace of its Jacobian matrix take the following forms:
\begin{equation}
    \text{Tr}(J(x^*, n^*)) = \frac{\theta}{(1 + \theta)^2} \frac{\partial f}{\partial x}(n^*),
    \label{27}
\end{equation}
\begin{equation}
    \text{Det}(J(x^*, n^*)) = -\varepsilon n^*(1 - n^*)\frac{\theta}{(1 + \theta)} \frac{\partial f}{\partial n}(x^*).
    \label{28}
\end{equation}

In particular, we focus on the determinant, for which it is clear that:

\begin{equation}
    \operatorname{sgn}(\text{Det}(J(x^*, n^*))) = \operatorname{sgn}\left(-\frac{\partial f}{\partial n}(x^*)\right).
    \label{29}
\end{equation}

Simplifying $\frac{\partial f}{\partial n}(x^*)$ yields:
\begin{equation}
\begin{split}
    &\frac{\partial f}{\partial n}(x^*) = \\
    &\frac{\theta d + c - (\theta b + a) + (1 + \theta)r}{\theta d + c} \cdot \left[d + \frac{1}{1 + \theta} \cdot (c - d)\right].
    \label{30}
\end{split}
\end{equation}
Combining this with the existence conditions for the internal equilibrium point, we have $\theta d + c - (\theta b + a) + (1 + \theta)r > 0$ and $d + \frac{1}{1 + \theta} \cdot (c - d) = \frac{\theta d + c}{1 + \theta} > 0$. Consequently, $- \frac{\partial f}{\partial n}(x^*) < 0$. So \(\text{Det}(J(x^*, n^*) \) is negative on its existence interval. Furthermore, $(x^*, n^*)$ is an unstable saddle point throughout its existence interval.
\end{IEEEproof}

To facilitate the summary of our results, Table~\ref{table1} presents the existence and stability conditions for the system’s equilibrium points when a punishment mechanism is imposed under resource-abundant conditions. In the following, we present specific numerical examples to verify the aforementioned theoretical analysis results when $a > b$ holds.

\textbf{Example 3.1} Numerical examples are provided to verify that the system exhibits bistability, where equilibrium points $(0,0)$ and $(1,1)$ are both stable as specified in Clause 2.3 of Theorem 2.

Figure~\ref{Figure1}(a) shows the phase diagram of the system. All four corner equilibrium points exist, and there is also an equilibrium point $(x_{n = 1}, 1)$ on the boundary. The internal equilibrium point $(x^*, n^*)$ exists but is unstable. Among these, $(0,0)$ and $(1,1)$ are stable. Fig.~\ref{Figure1}(c) shows the time evolution of the system’s states with initial conditions $x = 0.35$ and $n = 0.3$. The trajectory converges to $(0,0)$, predicting the occurrence of the tragedy of the commons. Fig.~\ref{Figure1}(d) shows the time evolution of the system’s states with initial conditions $x = 0.9$ and $n = 0.7$. The trajectory converges to $(1,1)$, verifying the stability of the ideal state in which all individuals adopt cooperative behavior and the environment remains in a state of abundance. Fig.~\ref{Figure1}(b) illustrates the attractor regions of the two equilibrium points mentioned above. The orange region converges to $(0,0)$, while the blue region converges to $(1,1)$. It can be seen that the attractor region of the ideal state is relatively small, requiring an extremely high initial level of cooperation and favorable environmental conditions to be reached.

\textbf{Example 3.2} Numerical illustration of the system's convergence to a bistable state at $(0,0)$ and $(x_{n=1},1)$ subject to the conditions in clause 2.4 of Theorem 2. 
\begin{figure}[!t]
    \centering
    \includegraphics[width=0.5\textwidth]{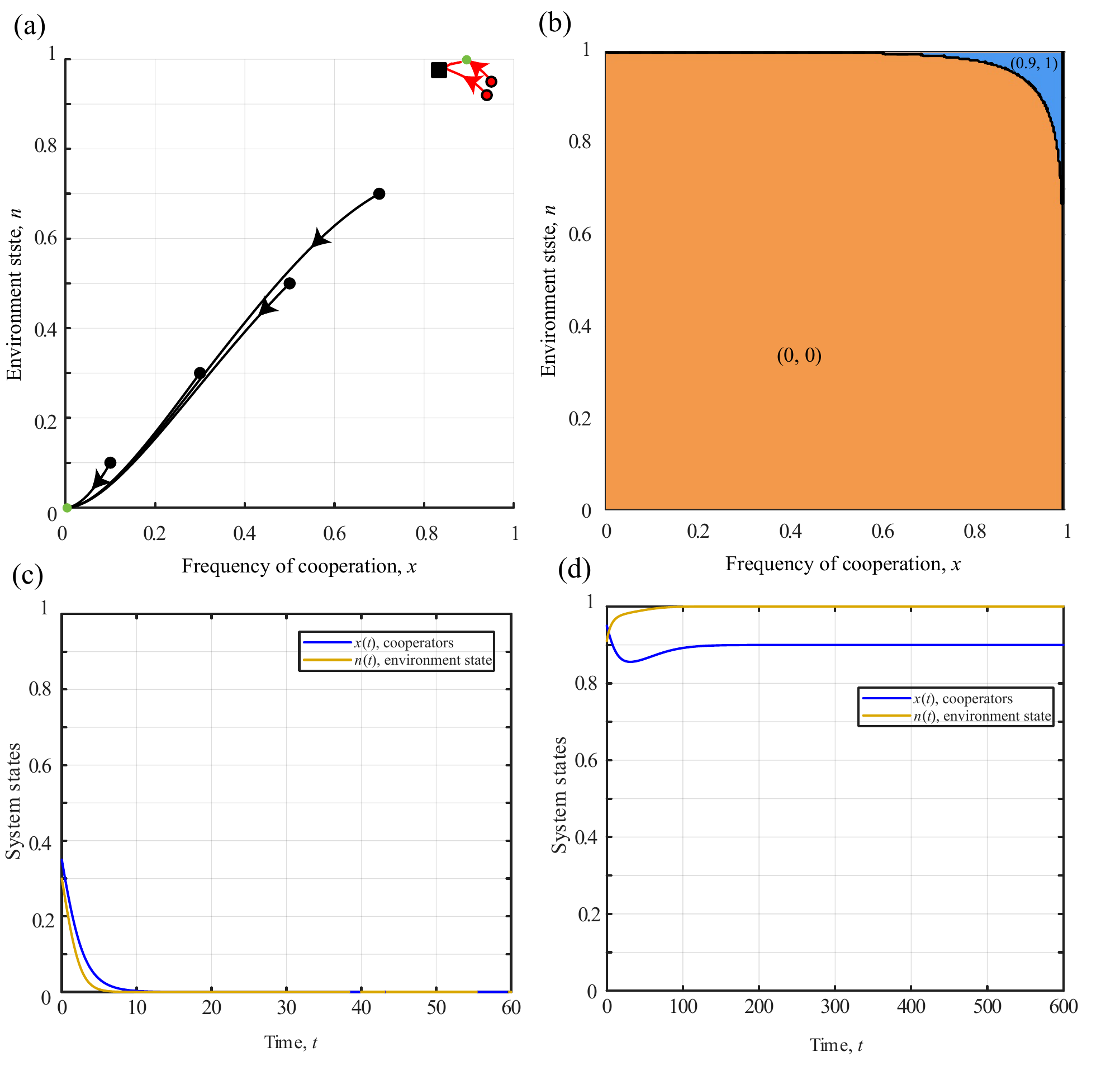}
\caption{Numerical verification of the stability of the boundary equilibrium $(x_{n=1}, 1)$ and the origin $(0,0)$ under clause 2.4 of Theorem 2.
Panel (a) shows the phase diagram, where $(0,0)$ and $(x_{n = 1}, 1)$ are stable points. Panel (b) illustrates the attractor regions of the two stable points, with the orange and blue regions representing the attractor regions of $(0,0)$ and $(x_{n = 1}, 1)$, respectively. Panel (c)  shows the time evolution of the system state with initial conditions $x=0.35$ and $n=0.3$, while Panel (d) shows the time evolution of the system state with initial conditions $x=0.95$ and $n=0.91$. Parameters are set as $R_0 = 3$, $S_0 = 2$, $T_0 = 4.5$, $P_0 = 2.5$, $R_1 = 3.5$, $S_1 = 2$, $T_1 = 5$, $P_1 = 3$, $\varepsilon = 1$, $\theta = 0.2$, $r = 1.45$.}
\label{Figure2}
\end{figure}

Figure~\ref{Figure2}(a) shows the phase diagram of the system. In this diagram, $(0,0)$ and $(x_{n = 1}, 1)$ are stable. Fig.~\ref{Figure2}(c) shows the time evolution of the system’s state with initial conditions $x = 0.35$ and $n = 0.3$. The trajectory converges to $(0,0)$, illustrating that the system eventually falls into the tragedy of the commons. Fig.~\ref{Figure2}(d) shows the time evolution of the system’s state with initial conditions $x = 0.95$ and $n = 0.91$. The trajectory converges to $(x_{n = 1}, 1)$, verifying the stability of boundary convergence, which signifies a resilient steady-state where stable cooperation maintains environmental abundance, effectively averting the tragedy of the commons. The partitioning of the phase space into basins of attraction is visualized in Fig.~\ref{Figure2}(b). Trajectories originating in the orange region gravitate toward $(0,0)$, in contrast to those in the blue region, which converge to the boundary equilibrium $(x_{n = 1}, 1)$. It is clear that the attraction basin for the boundary equilibrium is much smaller than the basin for $(0,0)$, which indicates the need for more demanding initial conditions.

The above example shows that when resources are abundant, if the temptation to defect is greater than the risk of being exploited ($a > b$), individuals tend to take advantage of the situation. At this point, when the punishment intensity $r$ is slightly higher than a critical value ($r > r_{c1}$), the system exhibits a bistable state: one attractor is the boundary state of ``partial cooperation but extremely abundant resources”, and the other is the tragedy of the commons state. When the punishment intensity is increased to the level of the temptation to defect ($r > a$), the system’s attractor optimizes toward a fully cooperative state. Similarly, when resources are abundant, and the temptation to defect exceeds the risk of exploitation ($a < b$), we can draw analogous conclusions. This suggests that imposing punishment in resource-rich environments helps prevent problems before they arise, suppresses partial defection, and allows the system the opportunity to evolve toward the ideal state of full cooperation.

\subsection{Evolutionary Dynamics Analysis of Punishment Under Resource Scarcity}

Subsequently, we analyze the evolutionary dynamics of the coupled system when punitive measures are imposed on defectors under conditions of resource scarcity. In this case, the system dynamics become highly complex as $r$ varies. First, we present the following theorem, which specifies the existence conditions for the system’s seven equilibrium points.

\begin{theorem} Apart from the corner equilibrium points $(0,0)$, $(0,1)$, $(1,0)$, and $(1,1)$, the system has an internal equilibrium point $(x^*, n^*)$ when $r > r_{c2} = \frac{\theta d + c}{1 + \theta}$. The existence conditions for the boundary equilibrium points $(x_{n = 0}, 0)$ depend on the magnitudes of $c$ and $d$. 
\begin{itemize}
\item when $c>d$, the existence condition is $d < r < c$; 
\item when $c<d$, the existence condition is $c < r < d$; 
\item $(x_{n = 1}, 1)$ does not exist in this system. Where:
\begin{equation}
   (x^*, n^*) = \left( \frac{1}{1 + \theta}, \frac{\theta d + c - (1 + \theta)r}{\theta d + c - (\theta b + a) - (1 + \theta)r} \right) ,
   \label{31}
\end{equation}
\begin{equation}
    (x_{n=0}, 0) = \left( \frac{r - d}{c - d}, 0 \right),
    \label{32}
\end{equation}
\begin{equation}
    (x_{n=1}, 1) = \left( \frac{b}{b - a}, 1 \right).
    \label{33}
\end{equation}
\end{itemize}
\end{theorem}
The proof of Theorem 3 follows a similar approach to that of Theorem 1, we will not elaborate further here.

The following theorem provides the stability conditions for the six equilibrium points.

\begin{theorem}
The stability of the six equilibrium points of the system satisfies the following conditions:

4.1 The stability of the corner points $(1,0)$ and $(1,1)$ is independent of $r$. In particular, $(0,1)$ and $(1,1)$ are both unstable saddle points.

4.2 When $r < d$, $(0,0)$ is a locally stable point.

4.3 The stability of the boundary equilibrium point $(x_{n=0},0)$ depends on the parameters $c$ and $d$.Assume that $c > d$ holds. Then, when $d < r < r_{c2}$, $(x_{n=0},0)$ is a locally stable point.

4.4 The stability of the internal equilibrium point $(x^*,n^*)$ depends on the parameters $a$, $b$, $c$, and $d$.

\begin{itemize}
\item Assuming $c > d$ and $a > b$, then $(x^*, n^*)$ is stable when $r > r_{c2}$;

\item Assuming $c > d$ and $a < b$, then $(x^*, n^*)$ is stable when $r_{c2} < r < r_{02}$;

\item Assuming $c < d$ and $a > b$, then $(x^*, n^*)$ is stable when $r > r_{02}$.
\end{itemize}
Here, $r_{02} = \frac{ad - bc}{a - b}$; the relationship between $r_{c2}$ and $r_{02}$ is influenced by the parameters $a$ and $b$.
\end{theorem}

\begin{IEEEproof}    
In this theorem, the proofs for 4.1–4.3 are similar to those in Theorem 2, so we will not repeat them here. Instead, we will focus on analysing the proof of 4.4.

First, we can determine the determinant and trace of the Jacobian matrix at an internal equilibrium point:
\begin{equation}
    \text{Tr}(J(x^*, n^*)) = \frac{\theta}{(1 + \theta)^2} \frac{\partial f}{\partial x}(n^*),
    \label{34}
\end{equation}
\begin{equation}
    \text{Det}(J(x^*, n^*)) = -\varepsilon n^*(1 - n^*) \frac{\theta}{1 + \theta} \frac{\partial f}{\partial n}(x^*).
    \label{35}
\end{equation}
Where:
\begin{equation}
    \frac{\partial f}{\partial x} = (1 + \theta)[n(b + r - d) + d - r],
    \label{36}
\end{equation}
\begin{equation}
    \frac{\partial f}{\partial n} = \frac{1}{n}[d - r + x(c - d)].
    \label{37}
\end{equation}
The analysis of the determinant is similar to that in Theorem 2, then we state the conclusion directly here: $\text{Det}(J(x^*, n^*))$ is positive on its domain of existence.

Next, we analyze the trace of the Jacobian matrix. Clearly, we have:
\begin{equation}
    \operatorname{sgn}(\text{Tr}(J(x^*, n^*))) = \operatorname{sgn}\left(\frac{\partial f}{\partial x}(n^*)\right).
    \label{38}
\end{equation}
Therefore, we can simplify $\frac{\partial f}{\partial x}(n^*)$ to:
\begin{equation}
    \frac{\partial f}{\partial x}(n^*) = (1 + \theta) \left[ \frac{(a - b)r + (bc - ad)}{\theta d + c - (\theta b + a) - (1 + \theta)r} \right].
    \label{39}
\end{equation}
Combining this with the existence conditions for $(x^*, n^*)$, we obtain $\theta d + c - (\theta b + a) - (1 + \theta)r < 0$. Therefore, we focus on the sign of $(a - b)r + (bc - ad)$. Clearly, this can be regarded as a continuous linear function of $r$, with a zero root at $r_{02} = \frac{ad-bc}{a-b}$. We now need to consider the relationship between the $r_{02}$ and $r_{c2}$. Taking the difference and simplifying yields:
\begin{equation}
    r_{c2} - r_{02} = \frac{\theta b + a}{1 + \theta} \cdot \frac{c - d}{a - b}.
    \label{40}
\end{equation}
Consequently, the relationship between the $r_{02}$ and $r_{c2}$ depends on the sign of $\frac{c-d}{a-b}$. Thus, we can comprehensively determine the sign of $\text{Tr}(J(x^*, n^*))$:

When $c > d$ and $a > b$, 
\begin{itemize}
\item if $r > r_{c2}$, then $\text{Tr}(J(x^*,n^*)) < 0$.  
\end{itemize}

When $c > d$ and $a < b$, 
\begin{itemize}
\item if $r_{c2} < r < r_{02}$, then $\text{Tr}(J(x^*,n^*)) < 0$; 
\item if $r > r_{02}$, then $\text{Tr}(J(x^*,n^*)) > 0$.
\end{itemize}

When $c < d$ and $a > b$, 
\begin{itemize}
\item if $r_{c2} < r < r_{02}$, then $\text{Tr}(J(x^*,n^*)) > 0$;  
\item if $r > r_{02}$, then $\text{Tr}(J(x^*,n^*)) < 0$.  
\end{itemize}

When $c < d$ and $a < b$, 
\begin{itemize}
\item if $r > r_{c2}$, then $\text{Tr}(J(x^*,n^*)) > 0$.  
\end{itemize}

Furthermore, by combining the signs of its Jacobian determinant and trace, we can derive the conclusion 4.4 in Theorem 4.
\end{IEEEproof} 

The following theorem specifies the conditions under which coupled systems exhibit other dynamical behaviors.

\begin{theorem} When a penalty is imposed under resource scarcity, the system exhibits heteroclinic\cite{NYCZKA2012317} or centrally stable dynamical behavior.

5.1 The conditions for the occurrence of a Hopf bifurcation depend on the signs of $d-c$ and $a-b$.
Assuming $d-c$ and $a-b$ have the same sign, which means $c > d$ and $a < b$ or $c < d$ and $a > b$ hold, the system exhibits a Hopf bifurcation at $r = r_{02}$ and can generate a family of periodic solutions in its vicinity.

5.2 The conditions for the emergence of heterocyclic cycle depend on the parameters $a$, $b$, $c$, and $d$.
\begin{itemize}
\item If $c > d$ and $a < b$, the condition for the emergence of a heteroclinic cycle is $r > r_{02}$;

\item If $c < d$ and $a > b$, the condition for the emergence of a heteroclinic cycle is $d < r < r_{02}$;

\item If $c < d$ and $a < b$, the condition for the emergence of a heteroclinic cycle is $r > d$.
\end{itemize}
\end{theorem}
\begin{IEEEproof}
 When the Jacobian matrix of an equilibrium point satisfies the conditions $\text{Tr}(J(x,n)) = 0$ and $\text{Det}(J(x,n)) > 0$, the system exhibits a Hopf bifurcation.

From the proof of Theorem 4, we can derive the content 5.1 of Theorem 5.

Next, we analyze the conditions for the emergence of a heteroclinic cycle. We can use dilemma strength to analyze the system’s behavior. When the system has no stable points, if the game structure of the internal equilibrium point $(x^*, n^*)$ is a Chicken game, the system will evolve toward a heteroclinic cycle. Based on the definitions of $D_g'$ and $D_r'$, we can express the dilemma strength of the internal equilibrium point as follows:
\begin{equation}
    D_r'^* = \frac{(a - b)r + bc - ad}{-[\theta b + a + (1 + \theta)f]r + f(\theta d + c) - e(\theta b + a)},
    \label{41}
\end{equation}
\begin{equation}
    \begin{split}
        &D_g'^* = \\
    &\frac{\theta[(b - a)r + ad - bc]}{-[\theta b + a + (1 + \theta)f]r + f(\theta d + c) - e(\theta b + a)} = -\theta D_r'^*,
    \end{split}
    \label{42}
\end{equation}
where $e = R_0 - P_0 > 0$ and $f = R_1 - P_1 > 0$.

Therefore, we need only analyze the sign of $D_r'^*$. Based on the existence conditions for an internal equilibrium point, we have $-[\theta b + a + (1 + \theta)f]r + f(\theta d + c) - e(\theta b + a) > 0$. For the numerator, we have a zero at $r_{02} = \frac{ad - bc}{a - b}$. Referring to the proof of Theorem 4, we can derive the content 5.2 of Theorem 5.
\end{IEEEproof}

To facilitate the summary of our results, TABLE~\ref{table2} presents the existence and stability conditions for the system’s equilibrium points when the punishment is imposed under resource scarcity.

Below, we provide some representative numerical examples to illustrate the more significant points from the discussion above. For further verification, please refer to the Appendix.

\textbf{Example 3.3} Numerical Verification of Boundary Convergence in the Clause 4.3 of Theorem 4.
\begin{figure}[!t]
    \centering
    \includegraphics[width=0.5\textwidth]{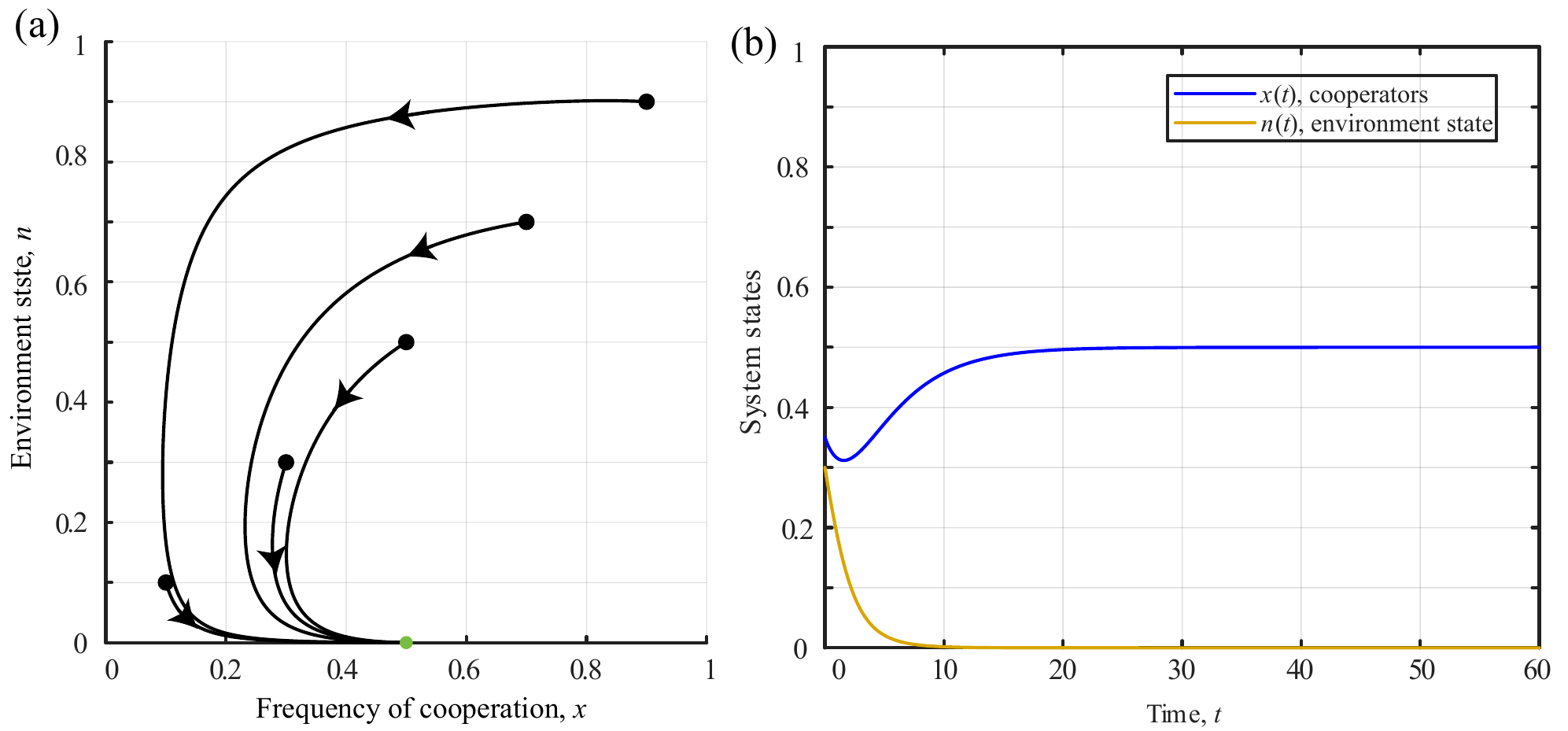}
\caption{Boundary convergence dynamics under penalty constraints in resource-constrained systems, subject to the conditions in clause 4.3 of Theorem 4. Panel (a) shows the phase diagram, where $(x_{n = 0}, 0)$ is a stable point. Panel (b) shows the time evolution of the system state with initial conditions $x = 0.35$ and $n = 0.3$. Parameters are set as $R_0 = 3$, $S_0 = 2$, $T_0 = 4.5$, $P_0 = 2.5$, $R_1 = 3.5$, $S_1 = 2$, $T_1 = 5$, $P_1 = 3$, $\varepsilon = 1$, $\theta = 0.2$, $r = 1$.}
\label{Figure3}
\end{figure}

Figure~\ref{Figure3}(a) shows the phase diagram of the system. The four corner equilibrium points all exist, and there is also an equilibrium point $(x_{n = 0}, 0)$ on the boundary. Among these, $(x_{n = 0}, 0)$ is stable. No internal equilibrium points exist within this parameter interval. Fig.~\ref{Figure3}(b) shows the time evolution with initial conditions $x = 0.35$ and $n = 0.3$. The trajectory converges to $(x_{n = 0}, 0)$, verifying the stability of this boundary equilibrium point.
This example demonstrates that when resources are scarce, if the temptation to defect exceeds the risk of exploitation ($c > d$), individuals are more inclined to actively defect. When the punishment intensity $r$ just exceeds the risk of exploitation $d$ ($r > d$), the system converges to a state of partial cooperation but complete resource depletion. This corresponds to a ``closing the barn door after the horse has bolted” type of intervention: while small penalties can prevent total collapse, they cannot restore the environment.

\textbf{Example 3.4} Numerical Verification of the Stability of the Internal Equilibrium Point in the second conclusion of clause 4.4 in Theorem 4.
\begin{figure}[!t]
    \centering
    \includegraphics[width=0.5\textwidth]{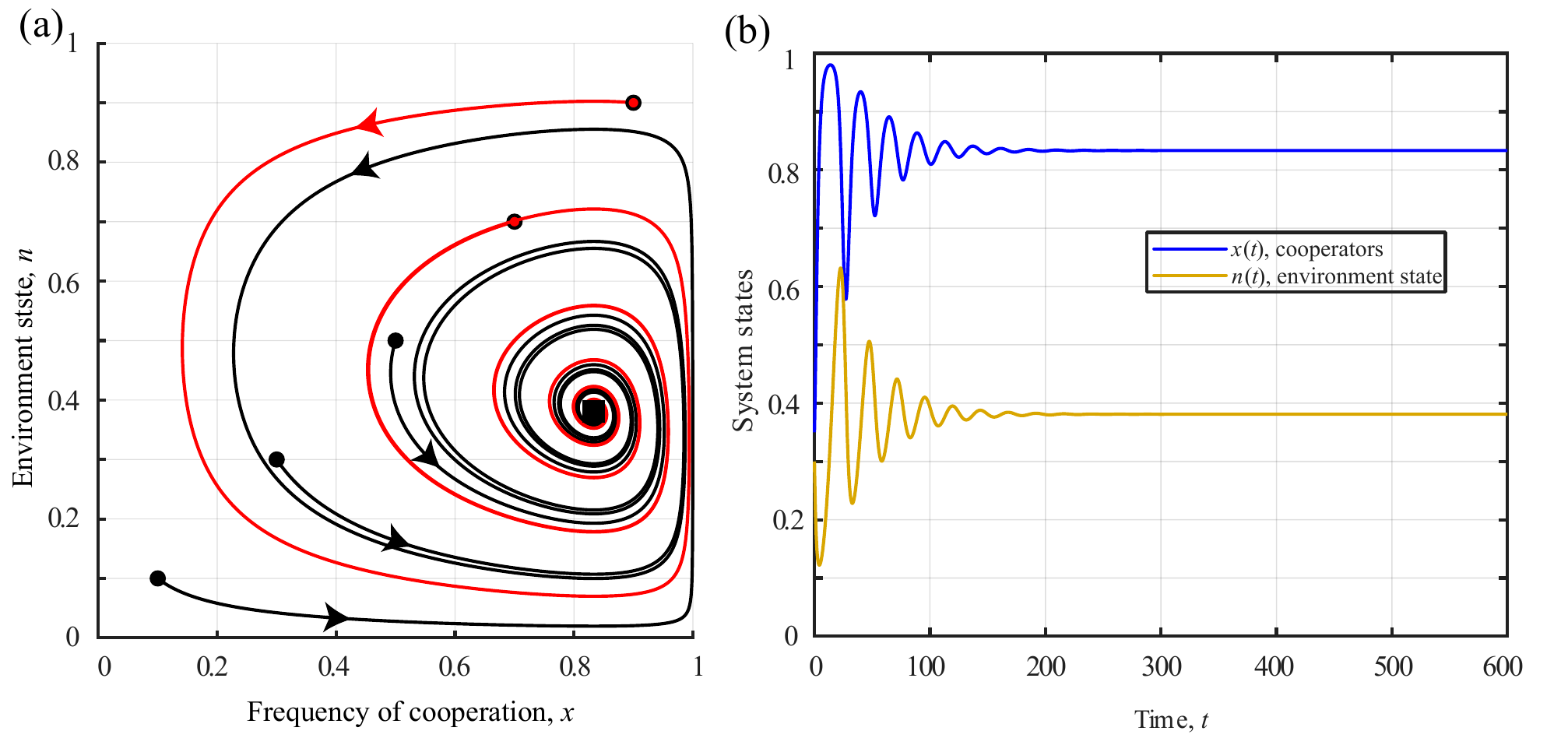}
\caption{Validating the internal stability under the second conclusion of clause 4.4 in Theorem 4 for systems characterized by penalty constraints and resource scarcity. Panel (a) shows the phase diagram, where $(x^*,n^*)$ is stable. Panel (b) shows the time evolution of the system state with initial conditions $x = 0.35$ and $n = 0.3$. Parameters are set as $R_0 = 3$, $S_0 = 2$, $T_0 = 4.5$, $P_0 = 2.5$, $R_1 = 3.5$, $S_1 = 1.5$, $T_1 = 4.5$, $P_1 = 3$, $\varepsilon = 1$, $\theta = 0.2$, $r = 2$.}
\label{Figure4}
\end{figure}

Figure~\ref{Figure4}(a) shows the phase diagram of the system. The internal equilibrium point $(x^*, n^*)$ is stable, and all trajectories converge to this point. Fig.~\ref{Figure4}(b) illustrates the time evolution starting from initial conditions $x = 0.35$ and $n = 0.3$, which converges to the internal equilibrium point, thereby verifying the theoretical prediction of the second conclusion of clause 4.4 in Theorem 4.

This example shows that if, during resource scarcity, the temptation to defect exceeds the risk of being exploited ($c > d$), and when resources improve, individuals fear exploitation more ($a < b$), then when the punishment intensity $r$ exceeds $a$ certain critical value, the system converges to a stable state where both cooperation and resource levels are moderate. This indicates that when individuals have a stronger defensive motivation under favorable conditions, punishment during resource scarcity is more likely to maintain system stability, even though it does not converge to the ideal state of full cooperation.
\begin{figure}[!t]
    \centering
    \includegraphics[width=0.5\textwidth]{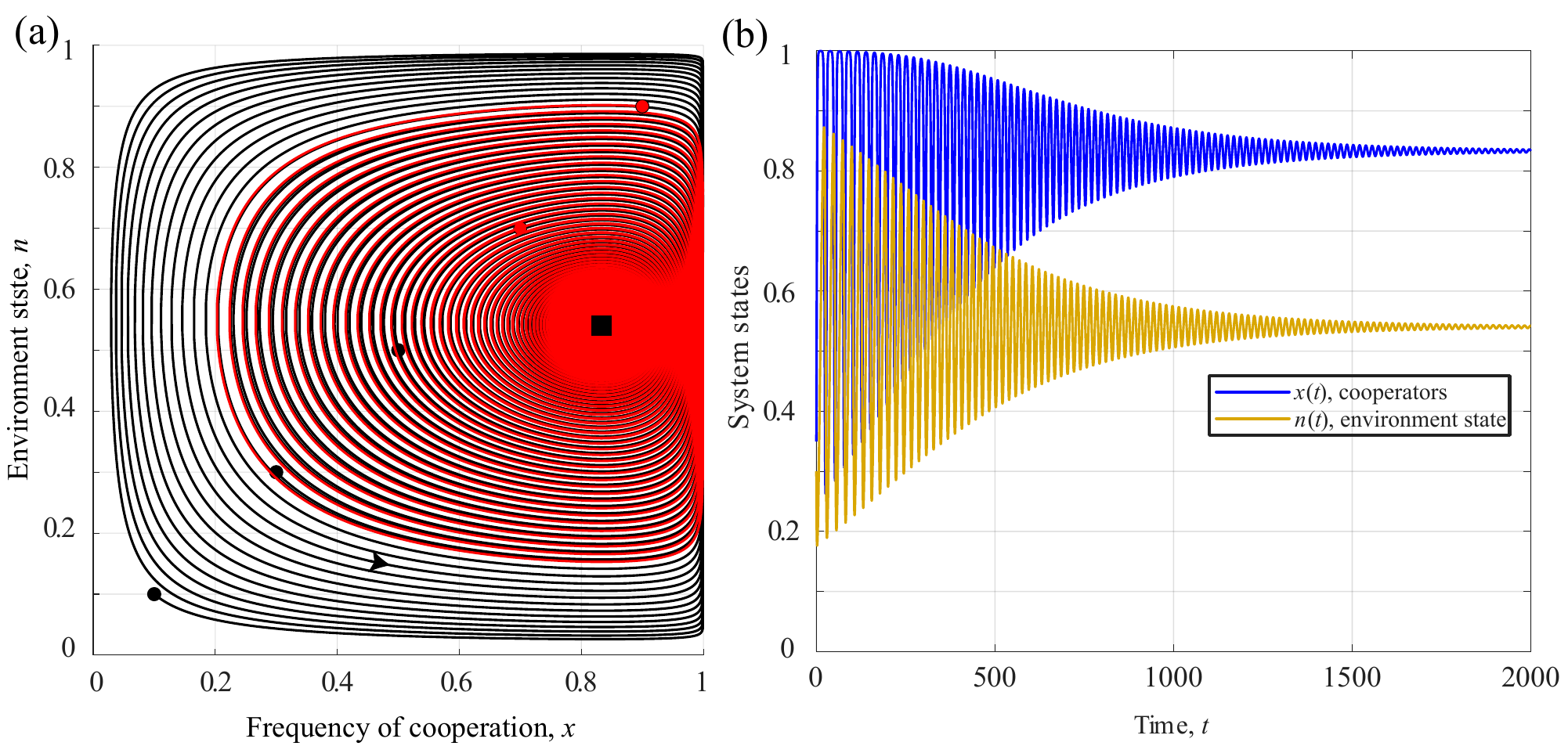}
\caption{Internal stable dynamics under resource scarcity with penalty when parameters meet the third condition of clause 4.4 in Theorem 4.
Panel (a) shows the phase diagram, where $(x^*,n^*)$ is stable. Panel (b)  shows the time evolution of the system state with initial conditions $x = 0.35$ and $n = 0.3$.  Parameters are set as $R_0 = 3.5$, $S_0 = 1.5$, $T_0 = 4.5$, $P_0 = 3$, $R_1 = 3$, $S_1 = 1.5$, $T_1 = 4.5$, $P_1 = 2.5$, $\varepsilon = 1$, $\theta = 0.2$, $r = 2.75$.}
\label{Figure5}
\end{figure}

As shown in Fig.~\ref{Figure5}, we further substantiate Theorem 4.4 through numerical analysis. The results indicate a critical dynamic: when risk-averse motives dominate during scarcity ($c < d$) and opportunistic defection prevails in abundance ($a > b$), the coupled system can still achieve stability. Specifically, provided that the defensive motivation in adverse conditions is strong enough, the implementation of stringent punitive measures serves as a stabilizer, effectively neutralizing the high temptation to defect in resource-rich scenarios."

\textbf{Example 3.5} Validating the Clause 5.1 of Theorem 5 for Periodic Closed Orbits within Hamiltonian Dynamics via Numerical Simulation.
\begin{figure}[!t]
    \centering
    \includegraphics[width=0.5\textwidth]{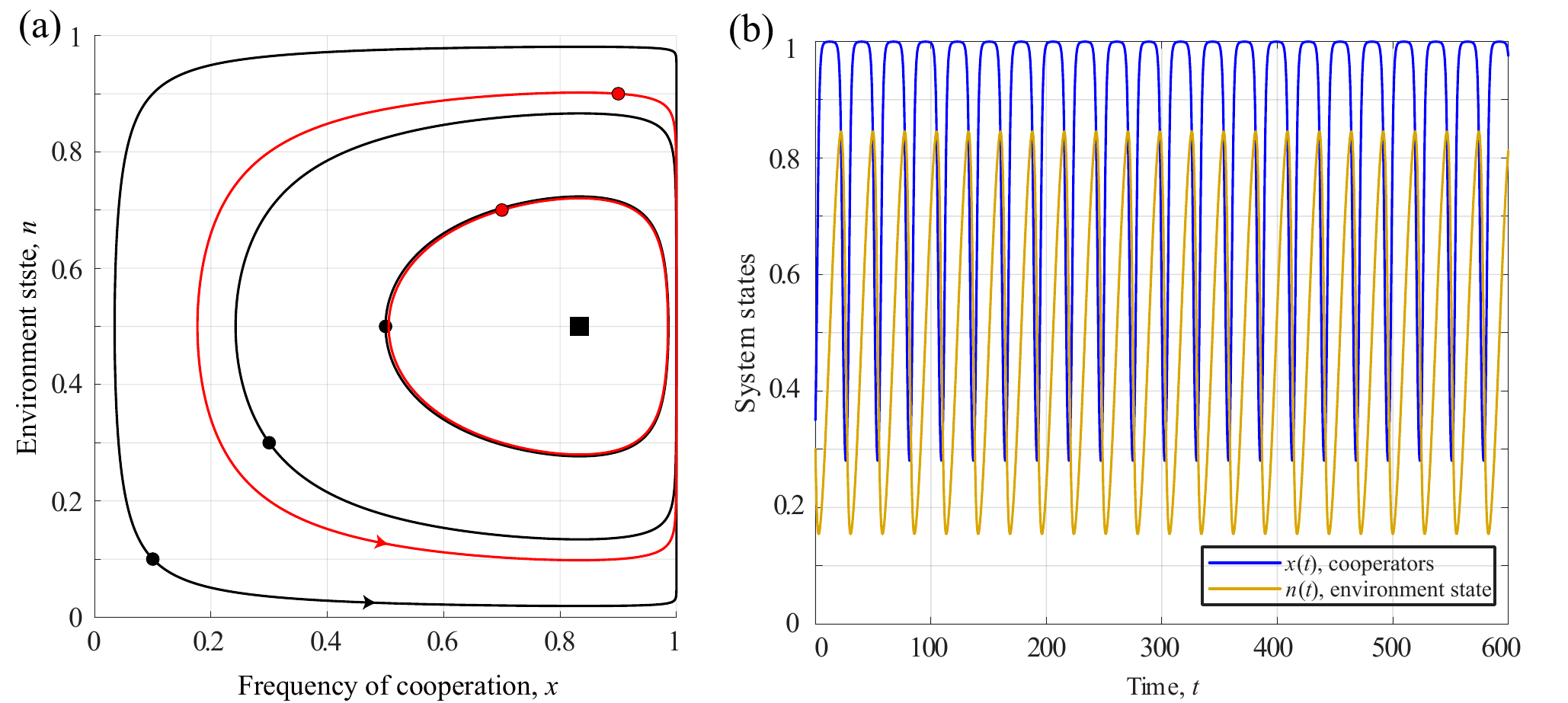}
\caption{Numerical evidence for central cycles in Hamiltonian systems as stated in the clause 5.1 of Theorem 5.
Panel (a) shows the phase diagram, in which all equilibrium points are unstable, and the evolution trajectory is periodic closed orbits around an internal equilibrium point. Panel (b) shows the time evolution of the system states with initial conditions x = 0.35 and n = 0.3. The parameters are set as $R_0 = 3.5$, $S_0 = 1.5$, $T_0 = 4.5$, $P_0 = 3$, $R_1 = 3$, $S_1 = 1.5$, $T_1 = 4.5$, $P_1 = 2.5$, $\varepsilon = 1$, $\theta = 0.2$, $r = 2.5$.}
\label{Figure6}
\end{figure}

Figure~\ref{Figure6}(a) shows the phase diagram near the bifurcation point. The system exhibits periodic closed orbits centered on an internal equilibrium point. Fig.~\ref{Figure6}(b) illustrates the time evolution for the initial conditions $x = 0.35$ and $n = 0.3$, where the cooperation frequency and environmental state exhibit sustained periodic oscillations.

This example demonstrates that when individuals tend to choose different strategies under favorable and unfavorable conditions, and a specific intensity of punishment is applied under adverse conditions, the system neither stabilizes at a single point nor collapses, but instead enters a predictable cyclical pattern. In practice, this requires policymakers to control the intensity of punishment with extreme precision, which is quite challenging to implement.

\textbf{Example 3.6} Numerical Illustration of Heteroclinic cycle to Substantiate Clause 5.2 of Theorem 5.

Figure~\ref{Figure7}(a) displays the phase diagram of the system, where all equilibrium points are unstable, and the heteroclinic cycle connects various saddle points to form a closed circuit. Fig.~\ref{Figure7}(c) shows the time evolution with initial conditions of $x = 0.35$ and $n = 0.3$, exhibiting irregular oscillations. Fig.~\ref{Figure7}(b) and Fig.~\ref{Figure7}(d) further illustrate the evolution of the penalty intensity, verifying that the system is in a heteroclinic loop state.

For numerical simulations of the system dynamics under other parameter settings, please refer to the Appendix. In light of the above considerations, the practical significance of implementing punishment mechanisms within resource-scarce environments can be analyzed as follows:
\begin{figure}[!t]
    \centering
    \includegraphics[width=0.5\textwidth]{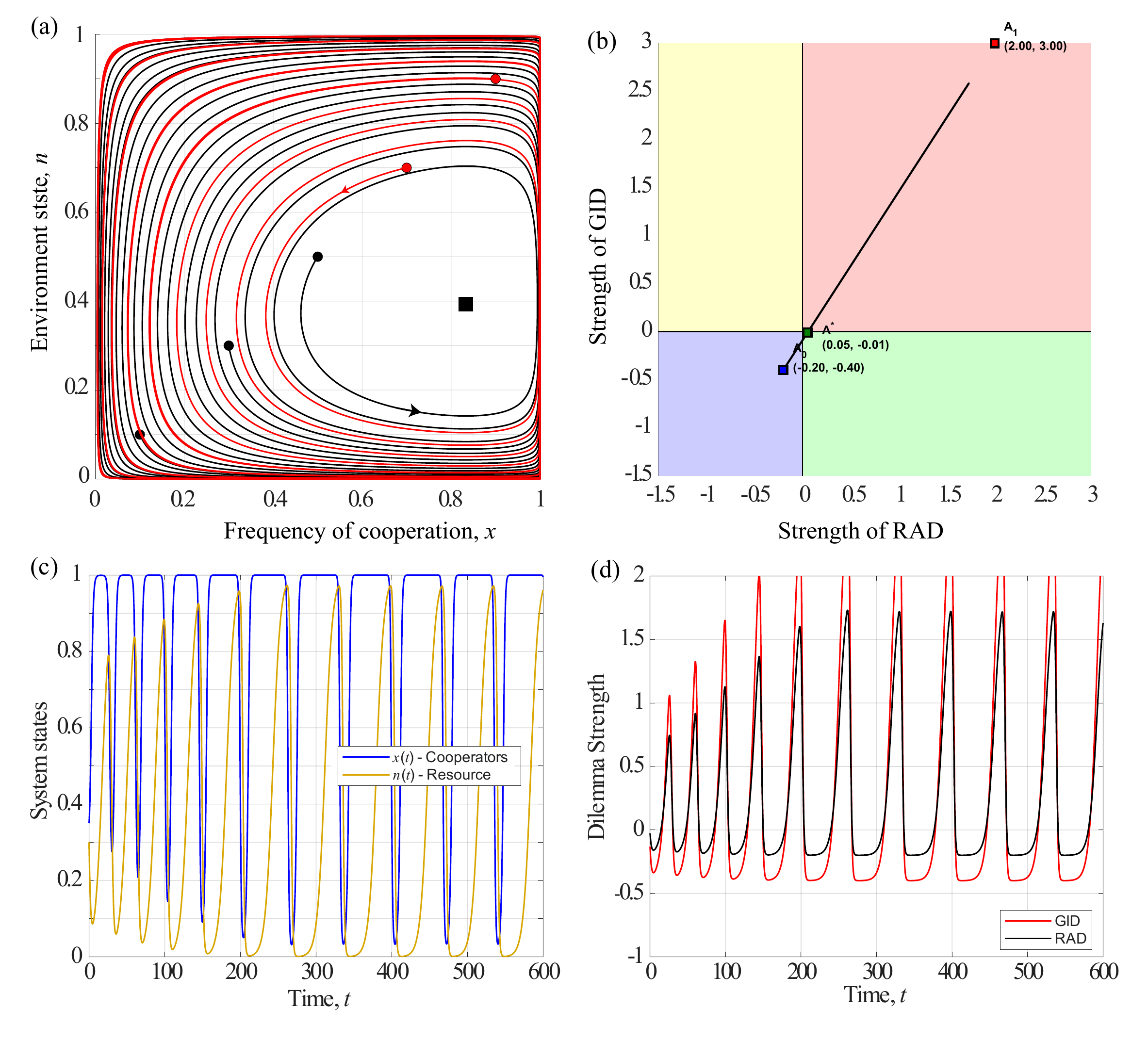}
\caption{Dynamical behavior of a heteroclinic cycle under resource scarcity with penalty satisfying clause 5.2 of Theorem 5. Panel (a) shows the phase diagram, in which all points are unstable, and the evolution trajectory forms an expanding heterocyclic cycle. Panel (c) shows the time evolution of the system state with initial conditions $x = 0.35$ and $n = 0.3$. Panel (b) shows the GID-RAD dilemma phase diagram. Panel (d) shows the time evolution of GID and RID. Parameters are set as $R_0 = 3.5$, $S_0 = 1.5$, $T_0 = 4.5$, $P_0 = 3$, $R_1 = 3$, $S_1 = 1.5$, $T_1 = 4.5$, $P_1 = 2.5$, $\varepsilon = 1$, $\theta = 0.2$, $r = 2$.}
\label{Figure7}
\end{figure}
\begin{table*}[!t]
\centering
\caption{Stability of the equilibrium point when penalties are imposed under resource scarcity}
\label{tab:scarce_combined}
\begin{tabular}{c c c}
\hline
Equilibrium points & Existence condition & Stability condition \\
\hline
$(0,0)$ & Always exist & Locally stable for $0<r<d$; saddle for $r>d$ \\
$(0,1)$ & Always exist & Unstable saddle \\
$(1,0)$ & Always exist & Unstable focus for $0<r<c$; saddle for $r>c$ \\
$(1,1)$ & Always exist & Unstable saddle \\
$(x_{n=0},0)$ & \begin{tabular}{c}
Exists when $c>d$ and $d<r<c$, \\
or $c<d$ and $c<r<d$
\end{tabular} & \begin{tabular}{c}
Case $c>d$: locally stable for $d<r<r_{c2}$, saddle for $r_{c2}<r<c$; \\
Case $c<d$: unstable focus for $c<r<r_{c2}$, saddle for $r_{c2}<r<d$; \\
\end{tabular} \\
$(x^{*},n^{*})$ & Exists when $r > r_{c2}$& \begin{tabular}{c}
Case $c>d$, $a>b$: locally stable for $r>r_{c2}$; \\
Case $c>d$, $a<b$: locally stable for $r_{c2}<r<r_{02}$, unstable focus for $r>r_{02}$; \\
Case $c<d$, $a>b$: unstable focus for $r_{c2}<r<r_{02}$, locally stable for $r>r_{02}$; \\
Case $c<d$, $a<b$: unstable focus for $r>r_{c2}$
\end{tabular} \\
\hline
\end{tabular}
\label{table2}
\end{table*}

When the temptation to defect surpasses the risk of exploitation across both environmental states ($c > d$ and $a > b$), a pervasive incentive for proactive defection emerges, severely undermining the foundation for cooperation. Under these circumstances, effective intervention necessitates a high punishment threshold ($r > r_{c2}$). If the punishment intensity is only moderate ($d < r < r_{c2}$), the system merely converges to a state of partial depletion; while total collapse is averted, the environment remains unable to achieve full recovery.

Conversely, if the temptation to defect dominates in degraded environments ($c > d$) while risk-aversion prevails in abundant ones ($a < b$), an 'optimal window' for punishment intensity ($r_{c2} < r < r_{02}$) exists, facilitating convergence to an internal equilibrium. However, this stability requires precise calibration: insufficient punishment precipitates a relapse into the Tragedy of the Commons, whereas excessive intensity induces oscillations within a heteroclinic cycle.

When individuals exhibit strong defensive motivations under resource scarcity ($c < d$) while facing higher temptations to defect in abundant environments ($a > b$), achieving system stability necessitates a high punishment intensity ($r > r_{02}$). In the absence of such stringent measures, the system is destined to collapse into either the Tragedy of the Commons or an unstable heteroclinic cycle.

Conversely, if a defensive mindset prevails regardless of the environmental state ($c < d$ and $a < b$), a robust foundation for collective cooperation is established. However, even in these inherently resilient societies, punishment alone may prove insufficient to extricate the system from systemic traps or oscillatory cycles. In such contexts, the imposition of external punishment should be discouraged; instead, alternative governance strategies, such as reputation-based mechanisms, should be prioritized to sustain long-term stability.

\section{Discussion}
In the context of public resource governance, punitive incentives have been proven to be an effective means of promoting cooperation and preventing the tragedy of the commons \cite{shijiagoverning,chen2014probabilistic,wang2025optimally}. However, previous research has largely focused on the intensity and form of incentives, while neglecting the question of under what conditions such incentives are most effective, namely, whether intervention should occur when conditions are still favorable or whether corrective measures should be taken after resources have already been degraded.

Here, within the framework of feedback-evolving games, we have constructed two coupled game models applicable to penalties imposed during resource abundance and penalties imposed during resource scarcity. We seek to identify pathways to cooperation within the classic and most challenging Prisoner’s Dilemma framework, and investigate under which conditions the application of penalties yields more significant effects. Using linear stability theory and classical numerical verification, we have systematically analyzed the equilibrium structure, stability conditions, and global dynamical behavior of the system under different timing of incentives. By employing analytical tools such as the dilemma phase space, we reveal how the intensity of incentives, game structure parameters, and environmental evolution parameters collectively shape the system’s long-term evolutionary trajectory, and we comprehensively classify and compare the differences in dynamical behavior under different timing of incentives.

Based on the research objectives, we have found that when punishment is introduced under conditions of resource abundance, the system has the potential to converge to an ideal state of full cooperation and resource abundance, although this possibility is typically accompanied by bistable system dynamics and requires high thresholds for both cooperation levels and environmental resource conditions. This finding suggests that imposing punishment during periods of resource abundance has the potential to shift the role of punitive incentives from treating symptoms to addressing root causes, thereby demonstrating the theoretical feasibility of preventive intervention strategies. Conversely, imposing punishment during periods of resource scarcity, while capable of rescuing the system from the nightmare of the tragedy of the commons, leads to complex and diverse dynamical behavior, potentially giving rise to phenomena such as centers and heterocyclic cycle. This creates significant challenges for policymakers in determining the appropriate intensity of penalties. Although we may be fortunate enough to avoid such problems and allow the system to reach a stable state, this stability is often accompanied by low resource levels, and the system’s final state is far from ideal. These findings provide clear guidance on the timing of punishment design in the governance of public resources: punishment should be applied when environmental conditions are still favorable to prevent problems before they arise, rather than waiting until resources are depleted to close the barn door after the horse has bolted.

In this study, we assume that the penalty intensity $r$ is a fixed constant, this assumption allows us to focus on the core question of when to apply incentives. Building on this foundation, future research could make the penalty intensity dynamic, allowing it to adaptively adjust based on the environmental state $n$ and the level of cooperation $x$, thereby constructing environment-dependent or cooperation-dependent dynamic punishment mechanisms \cite{hua2024coevolutionary}. Exploring the timing of application within this framework holds greater practical significance. Furthermore, we assume that the punishment mechanism takes effect immediately, whereas in the real world, various delays may exist \cite{zheng2021modeling,zhao2016matrix,hu2023evolutionary}. On the one hand, there may be a time lag between the formulation of punishment systems and the implementation of policies. On the other hand, there may also be a time lag between the implementation of systems and the achievement of governance outcomes. Analyzing the effects of these time lags helps us better approximate real-world governance scenarios and provides policymakers with more actionable theoretical guidance.

\appendix
\section*{Detailed numerical examples}
In this appendix, we provide all the numerical examples used in this study.
\subsection*{Numerical Example of Punishment Under Resource Abundance}

We have already provided numerical examples of how to escape the tragedy of the commons by imposing penalties in an environment of abundant resources. The following two examples demonstrate that when the penalty is too weak, the system cannot escape the tragedy of the commons.
\begin{figure}
    \centering
    \includegraphics[width=0.5\textwidth]{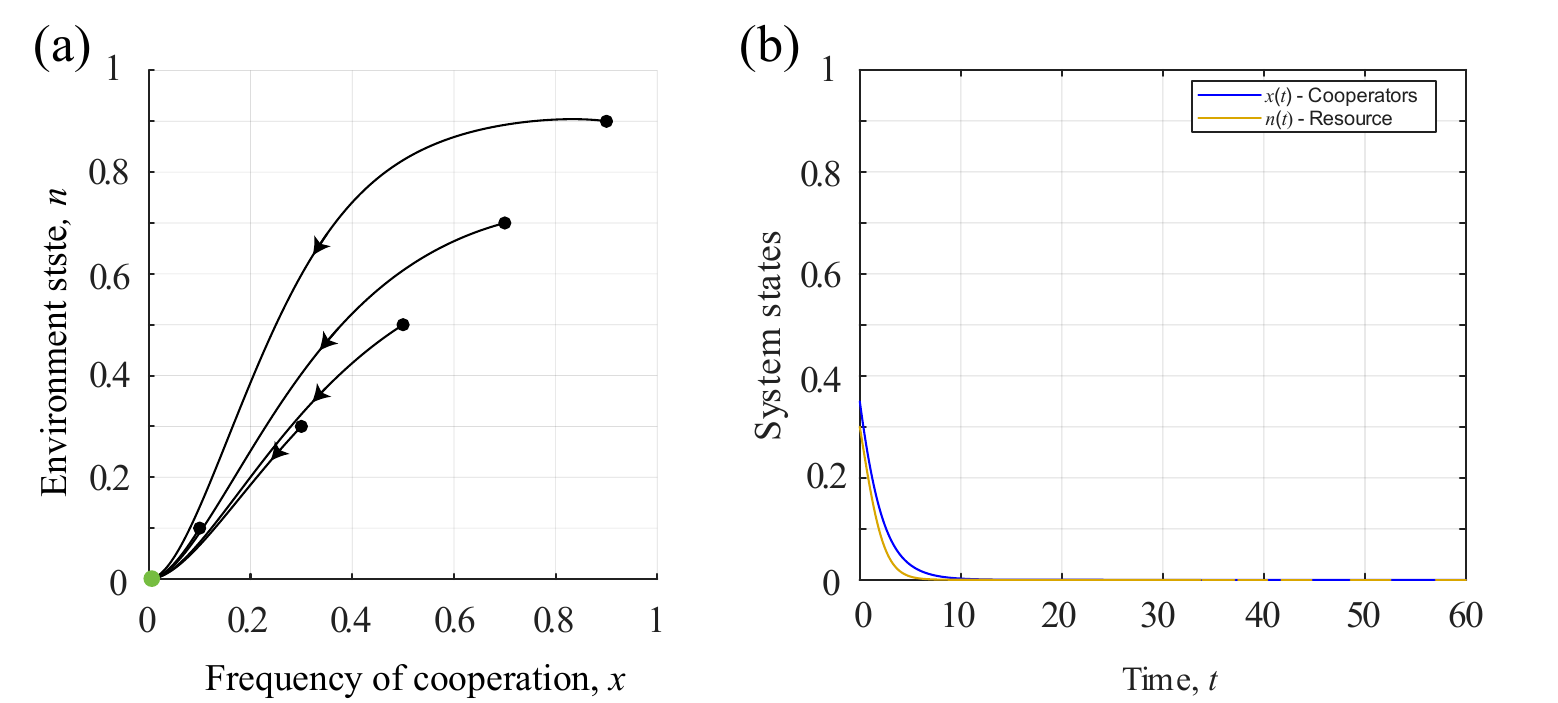}
\caption{The tragedy of the commons with punishment under abundant resources.
Panel (a) shows the phase diagram, where $(0,0)$ is a stable point. Panel (b) shows the time evolution of the system state with initial conditions $x = 0.35$ and $n = 0.3$. The trajectory converges to $(0,0)$. Parameters are set as $R_0 = 3$, $S_0 = 2$, $T_0 = 4.5$, $P_0 = 2.5$, $R_1 = 3.5$, $S_1 = 2$, $T_1 = 5$, $P_1 = 3$, $\varepsilon = 1$, $\theta = 0.2$, $r = 1$.}
\label{Figure8}
\end{figure}

\begin{figure}
    \centering
    \includegraphics[width=0.5\textwidth]{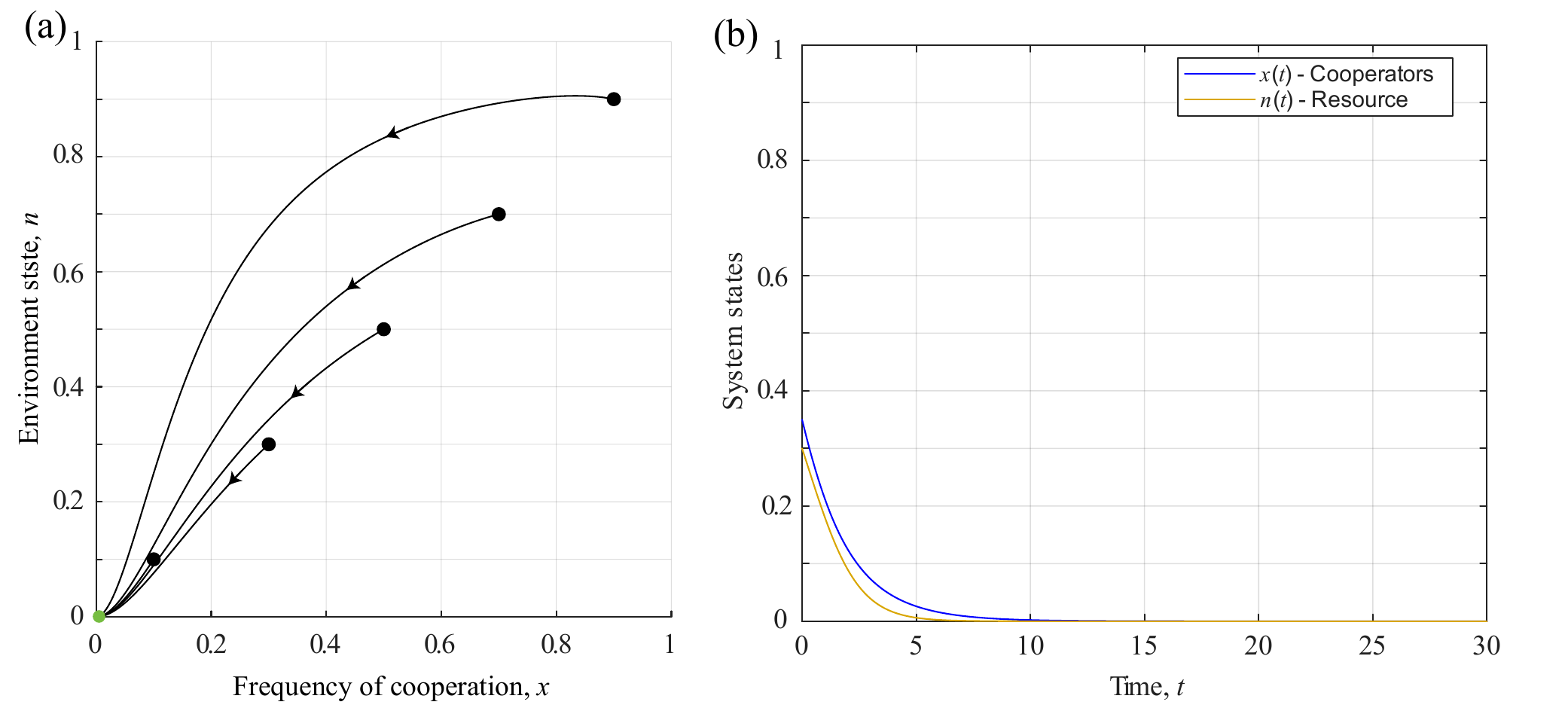}
\caption{The tragedy of the commons with punishment under abundant resources.
Panel (a) shows the phase diagram, where $(0,0)$ is a stable point. Panel (b) shows the time evolution of the system state with initial conditions $x = 0.35$ and $n = 0.3$; the trajectory converges to $(0,0)$. Parameters are set as $R_0 = 3$, $S_0 = 2$, $T_0 = 4.5$, $P_0 = 2.5$, $R_1 = 3.5$, $S_1 = 2$, $T_1 = 4.5$, $P_1 = 3.25$, $\varepsilon = 1$, $\theta = 0.2$, $r = 0.75$.}
\label{1.2}
\end{figure}

Under conditions of environmental abundance, if the punitive fine remains below $T_1-R_1$, the payoff advantage of choosing defection over cooperation against a cooperative counterpart, the corner equilibrium $(0,0)$ attains stability, consistent with Clause 2.2 of Theorem 2. As shown in Figs.~\ref{Figure8} and ~\ref{1.2}, the coupled system possesses a unique stable equilibrium at $(0,0)$, which signifies a scenario where all individuals adopt defecting behavior within a resource-depleted environment.

\subsection*{Numerical Example of Punishment Under Resource Scarcity}

As mentioned earlier, the dynamics of punishment under resource scarcity are quite complex. Here we present the remaining numerical examples below, categorized by scenario.

\textbf{Case 1: The tragedy of the commons}

As demonstrated in Fig.~\ref{Figure10}, under certain conditions of resource scarcity, the punishment intensity $r$ may prove insufficient to offset the temptation to defect. The phase portrait in Fig.~\ref{Figure10}(a) and the corresponding time series in Fig.~\ref{Figure10}(b) reveal that the corner equilibrium $(0,0)$ acts as the unique stable attractor. In this scenario, regardless of the initial level of cooperation or environmental state, the system inevitably gravitates toward total resource depletion and universal defection. 

\begin{figure}
    \centering
    \includegraphics[width=0.5\textwidth]{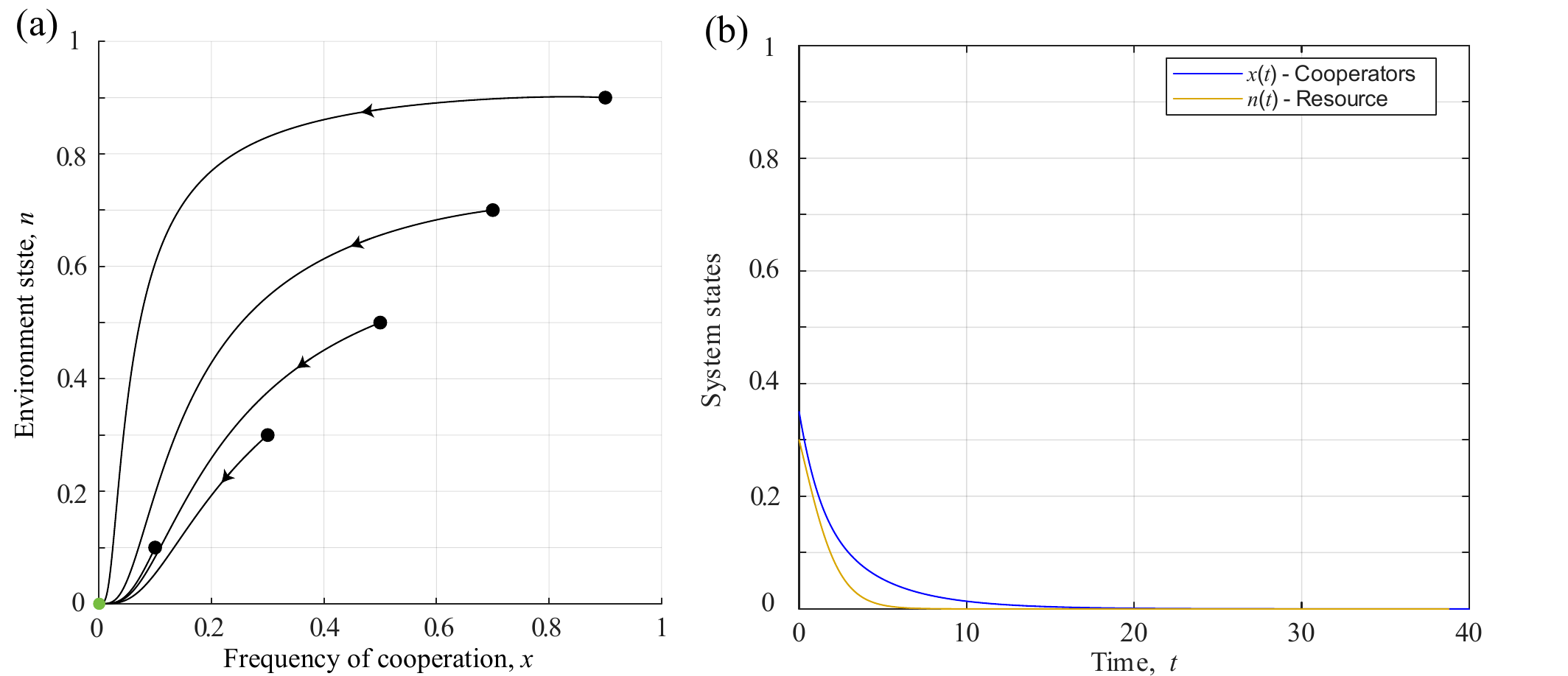}
\caption{The tragedy of the commons with punishment under resource scarcity.
Panel (a) shows the phase diagram, where $(0,0)$ is a stable point. Panel (b) shows the time evolution of the system state with initial conditions $x = 0.35$ and $n = 0.3$. The trajectory converges to $(0,0)$. Parameters are set as $R_0 = 3$, $S_0 = 2$, $T_0 = 4.5$, $P_0 = 2.5$, $R_1 = 3.5$, $S_1 = 2$, $T_1 = 5$, $P_1 = 3$, $\varepsilon = 1$, $\theta = 0.2$, $r = 0.25$.}
\label{Figure10}
\end{figure}

\textbf{Case 2: Periodic closed orbit}

While the previous analysis focused on the scenario $c < d$ and $a > b$, we now examine the alternative case as illustrated in Fig. \ref{Figure6}. Specifically, we consider the regime where $c > d$ and $a < b$.

The numerical analysis presented in Fig. \ref{Figure11} confirms a distinct dynamical regime characterized by central periodic behavior, as analytically predicted in the clause 5.1 of Theorem 5. As illustrated in Fig. \ref{Figure11}(a), the phase portrait reveals that the interior equilibrium point is surrounded by a family of periodic closed orbits. Unlike the convergence to boundary or corner equilibria observed in previous cases, the trajectories here circulate indefinitely. The corresponding time series in Fig. \ref{Figure11}(b) further substantiates this dynamic. Both the frequency of cooperation $x(t)$ and the environmental state $n(t)$ exhibit synchronized, sustained oscillations. When environmental resources increase, the incentive for cooperation eventually wanes, leading to a decline in resources. Subsequently, the resulting scarcity triggers the punishment mechanism and the defensive motivations of individuals, which restores cooperation and initiates a new cycle of environmental recovery.

\begin{figure}
    \centering
    \includegraphics[width=0.5\textwidth]{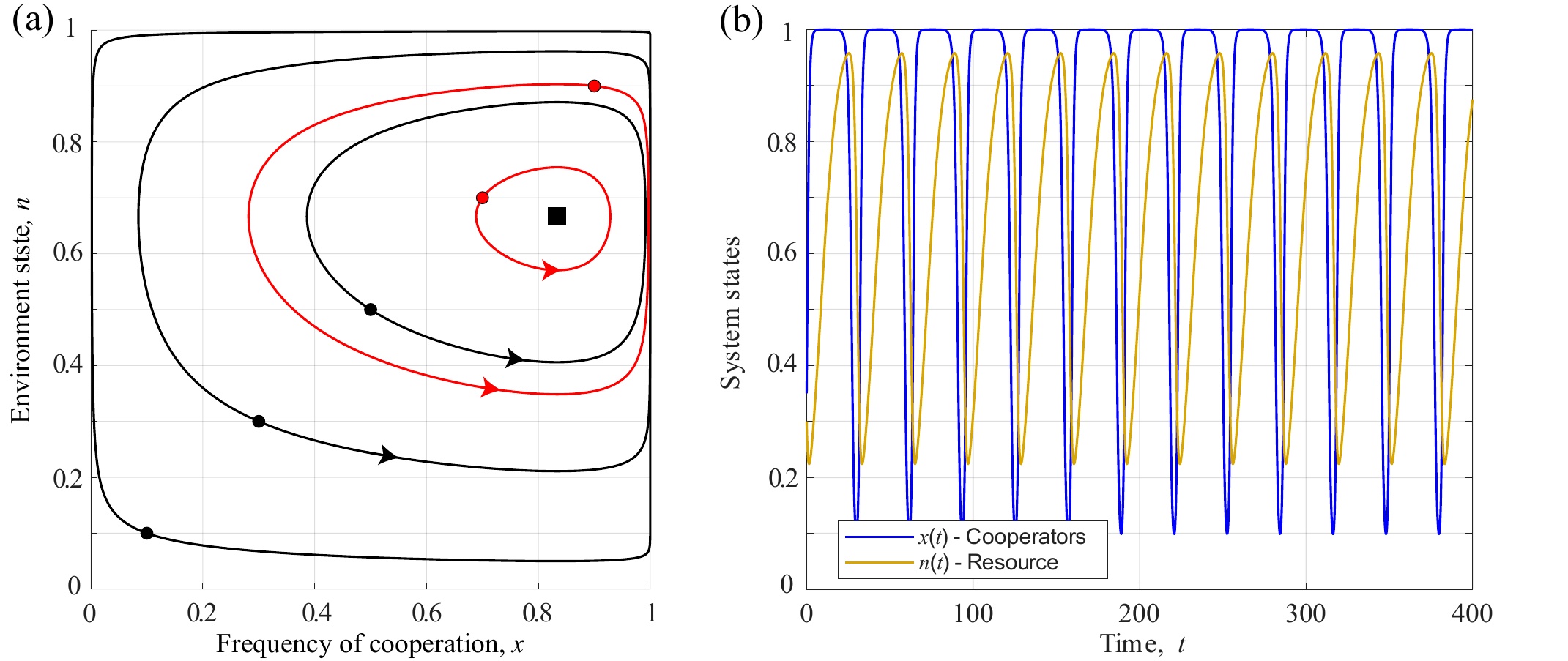}
\caption{Central periodic behavior under resource scarcity with penalty.
Panel (a) The figure shows the phase diagram, where all points are unstable, and the evolution trajectory is a stable limit cycle around the internal equilibrium point. Panel (b) shows the time evolution of the system state with initial conditions $x = 0.35$ and $n = 0.3$. The parameters are set as $R_0 = 3$, $S_0 = 2$, $T_0 = 4.5$, $P_0 = 2.5$, $R_1 = 3.5$, $S_1 = 1.5$, $T_1 = 4.5$, $P_1 = 3$, $\varepsilon = 1$, $\theta = 0.2$, $r = 3.5$.}
\label{Figure11}
\end{figure}

\textbf{Case 3: Heteroclinic cycle}

In Fig.\ref{Figure7}, we illustrated the emergence of heteroclinic cycle dynamics under the conditions $c < d$ and $a > b$, with the punishment intensity falling within the range $d < r < r_{02}$. To further generalize these findings, we now present the heteroclinic cycle oscillations under two additional parameter regimes, as specified in the first and third clauses of 5.2 in Theorem 5.
\begin{figure}
    \centering
    \includegraphics[width=0.5\textwidth]{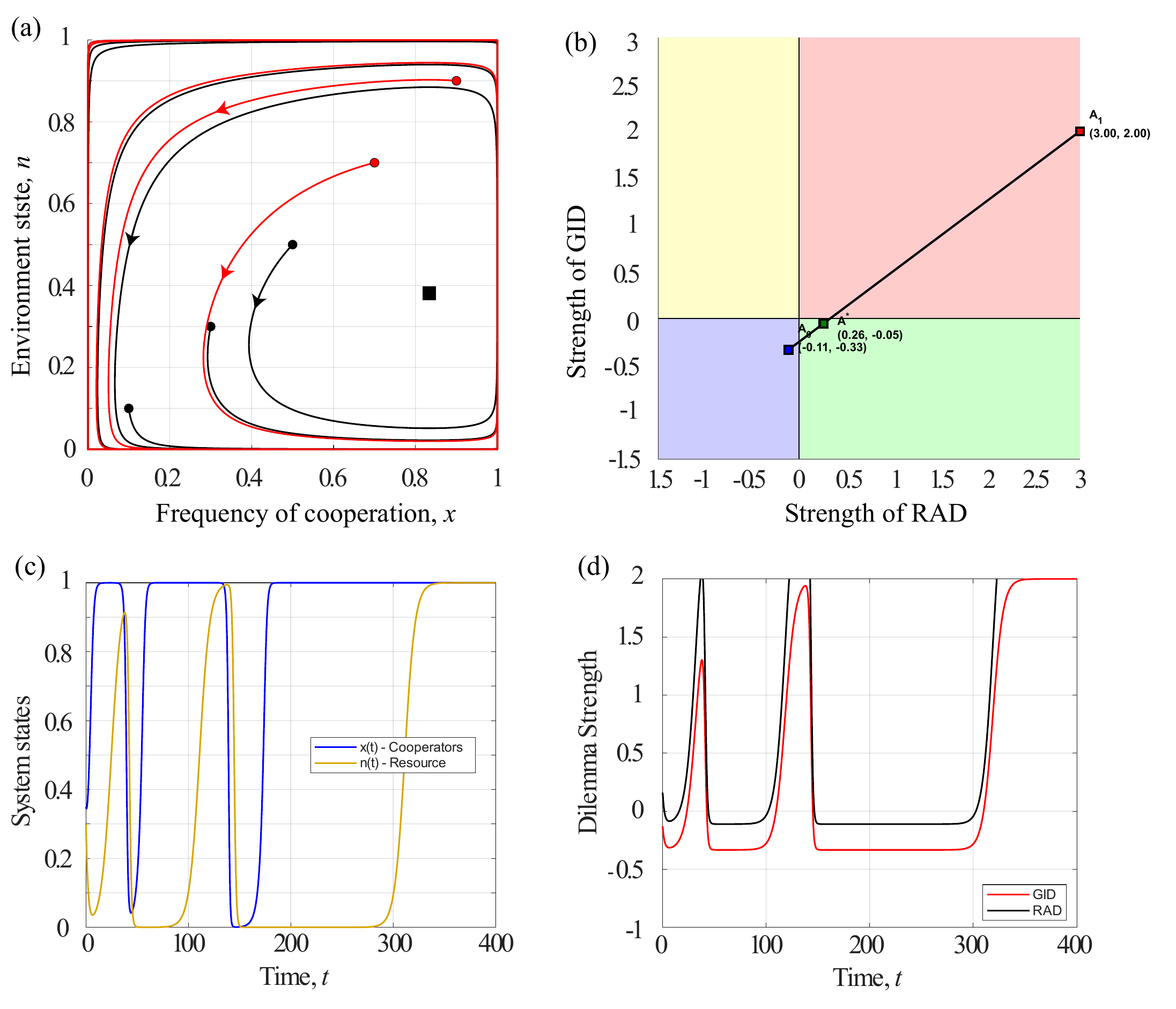}
\caption{Dynamical behavior of a heteroclinic cycle under resource scarcity with puinshment. Panel (a) shows the phase diagram, in which all points are unstable, and the evolution trajectory forms an expanding heterocyclic cycle. Panel (c) shows the time evolution of the system states with initial conditions $x = 0.35$ and $n = 0.3$. Panel (b) shows the GID-RAD dilemma phase diagram. Panel (d) shows the time evolution of GID and RID. Parameters are set as $R_0 = 3.5$, $S_0 = 1.5$, $T_0 = 4.5$, $P_0 = 3$, $R_1 = 3.5$, $S_1 = 1.5$, $T_1 = 4.5$, $P_1 = 3$, $\varepsilon = 1$, $\theta = 0.2$, $r = 1.75$.}
\label{Figure12}
\end{figure}
As illustrated in Fig.\ref{Figure12}, when $a < b$, $c < d$ and the punishment intensity is relatively high ($r > d$), the coupled system exhibits an expanding heteroclinic cycle. This dynamical regime validates the theoretical predictions in the third clause of 5.2 in Theorem 5. The phase portrait in Fig.\ref{Figure12}(a) shows that the evolutionary trajectories do not converge to a stable internal or boundary equilibrium. Instead, they form a heteroclinic cycle that connects the four corner equilibria of the system. As time progresses, the trajectories spiral outward, approaching the boundaries ever more closely. The time-series plot in Fig.\ref{Figure12}(c) further characterizes this behavior. The results in Fig.\ref{Figure12}(b) and (d) confirm that the heteroclinic cycle is propelled by the dynamic coupling of game payoffs and environmental states. Specifically, periodic spikes in dilemma strengths drive the rapid transitions between cooperation and defection observed in the phase space.

\begin{figure}
    \centering
    \includegraphics[width=0.5\textwidth]{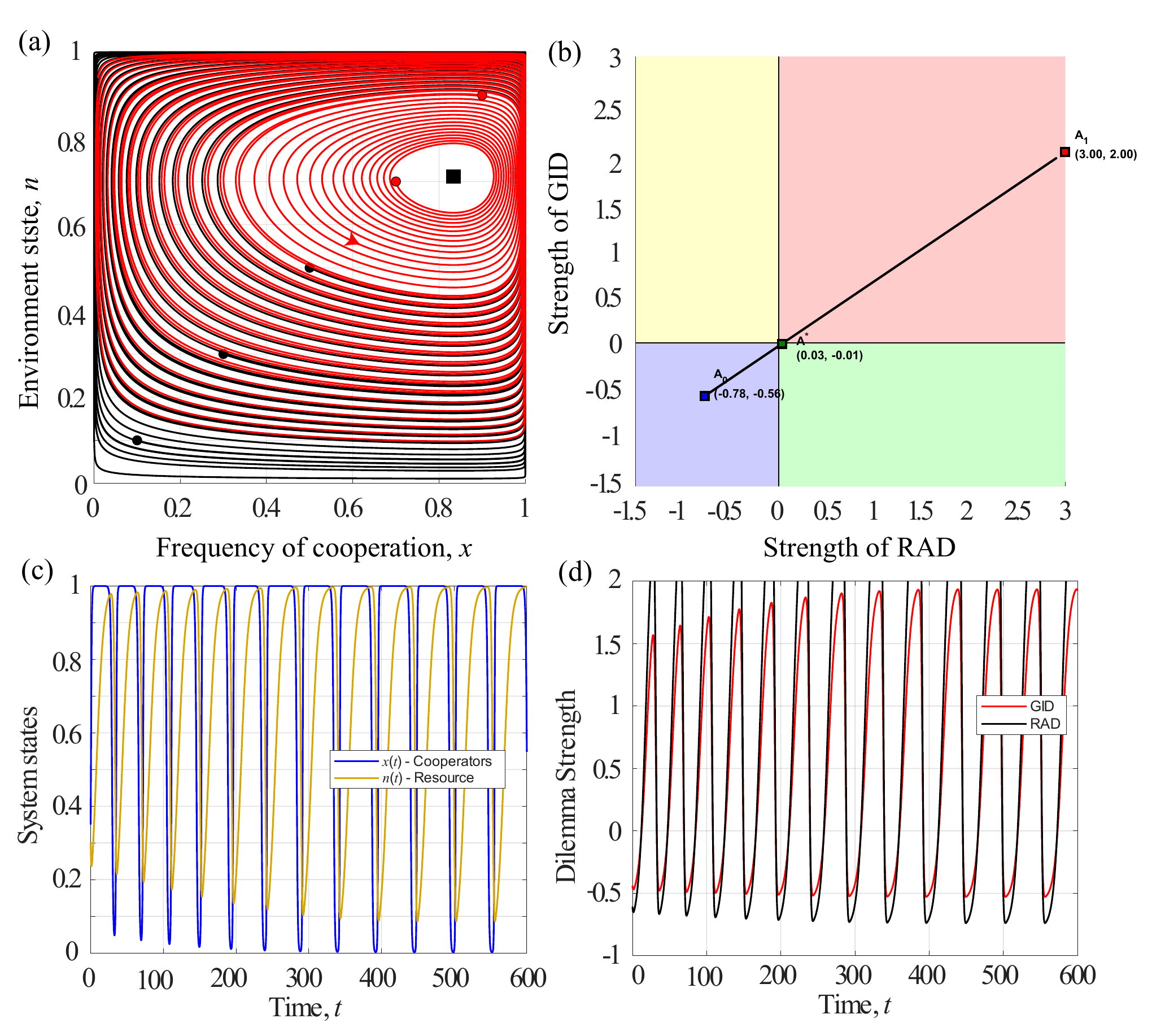}
\caption{Dynamical behavior of a heteroclinic cycle under resource scarcity with penalty. Panel (a) shows the phase diagram, in which all points are unstable, and the evolution trajectory forms an expanding heterocyclic loop. Panel (c) shows the time evolution of the system state with initial conditions x = 0.35 and n = 0.3. Panel (b) The figure shows the GID-RAD dilemma phase diagram. Panel (d) The figure shows the time evolution of GID and RID. Parameters are set as $R_0 = 3$, $S_0 = 2$, $T_0 = 4.5$, $P_0 = 2.5$, $R_1 = 3.5$, $S_1 = 1.5$, $T_1 = 4.5$, $P_1 = 3$, $\varepsilon = 1$, $\theta = 0.2$, $r = 4$.}
\label{Figure13}
\end{figure}
Next, we examine the dynamical regime where $c > d$, $a < b$, and $r > r_{02}$, as specified in the  first clause of 5.2 in Theorem 5. The numerical results presented in Fig.\ref{Figure13} confirm the emergence of an expanding heteroclinic cycle. As illustrated in Fig.\ref{Figure13}(a), the phase portrait reveals that even with a high punishment intensity ($r > r_{02}$), the system fails to reach a stable internal equilibrium. Instead, the evolutionary trajectories form an expanding heteroclinic cycle that sequentially approaches the four corner saddle points. This indicates a state of extreme instability where the system is perpetually driven toward the boundaries of the state space. The temporal evolution in Fig.\ref{Figure13}(c) further characterizes this behavior through pulse-like oscillations. Regarding the strategic drivers, Fig.\ref{Figure13}(b) and (d) illustrate that the heteroclinic cycle is propelled by the co-evolutionary feedback between environmental states and game payoffs.

\textbf{Case 4: Stable interior equilibrium}
\begin{figure}
    \centering
    \includegraphics[width=0.5\textwidth]{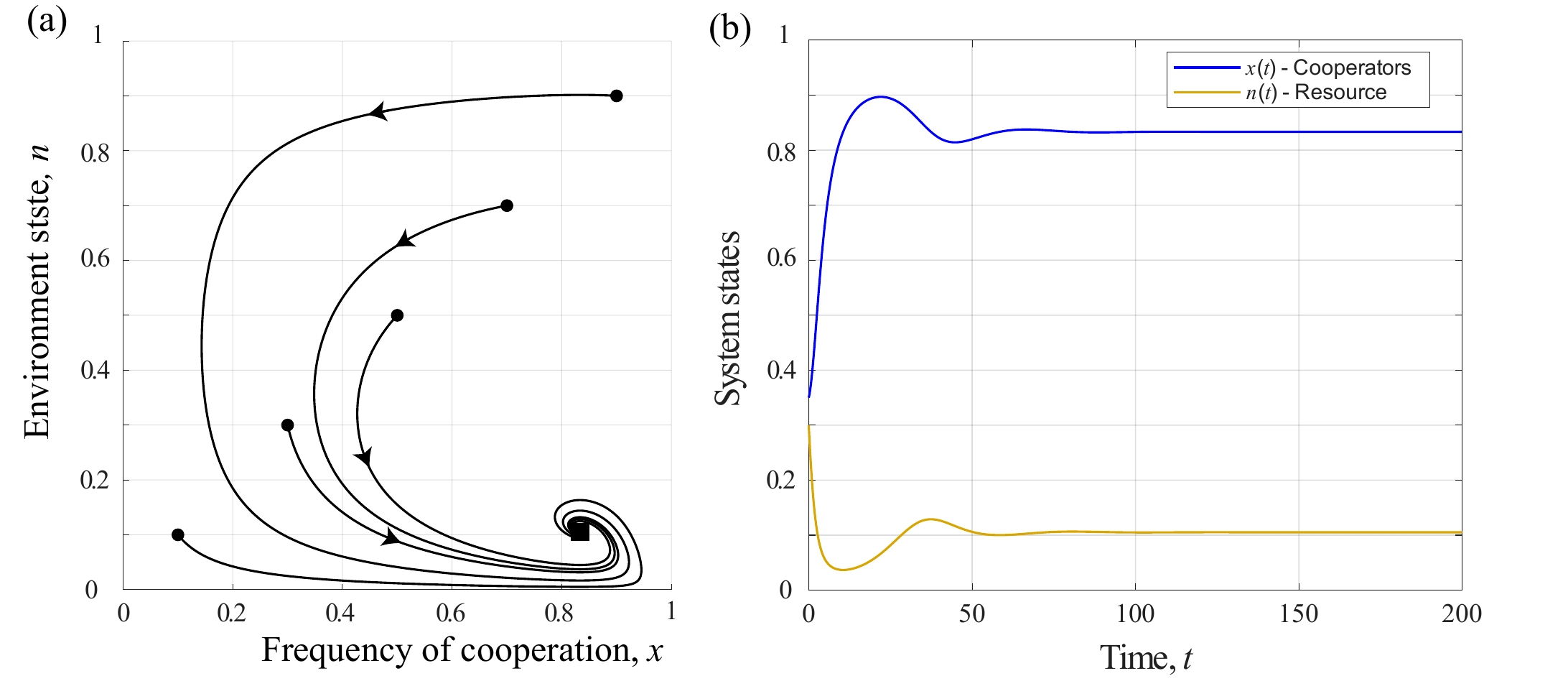}
\caption{The tragedy of the commons with punishment under resource scarcity.
Panel (a) shows the phase diagram, where $(0,0)$ is a stable point. Panel (b) shows the time evolution of the system state with initial conditions $x = 0.35$ and $n = 0.3$; the trajectory converges to $(0,0)$. Parameters are set as $R_0 = 3$, $S_0 = 2$, $T_0 = 4.5$, $P_0 = 2.5$, $R_1 = 3.5$, $S_1 = 2$, $T_1 = 5$, $P_1 = 3$, $\varepsilon = 1$, $\theta = 0.2$, $r = 1.5$.}
\label{Figure14}
\end{figure}

Following the analysis of the second and third conclusions in 4.4 of Theorem 4 (shown in Figs.\ref{Figure4} and \ref{Figure5}), we now provide a numerical illustration for the first conclusion of the same theorem. The results presented in Fig.\ref{Figure14} verify the emergence of an interior stable equilibrium point when $a < b$ and $c < d$ and the punishment intensity is sufficiently high ($r > r_{c2}$).

As shown in Fig.\ref{Figure14}(a), the phase portrait illustrates that the evolutionary trajectories, regardless of their initial configurations, eventually converge to a unique interior stable equilibrium point. The system reaches a state where human cooperation and environmental resources coexist at sustainable, non-zero levels. The temporal evolution depicted in Fig.\ref{Figure14}(b) further confirms this stability. After a brief period of dampened oscillations, both the frequency of cooperation $x(t)$ and the environmental state $n(t)$ stabilize at a constant steady state.


\bibliographystyle{IEEEtran}
\bibliography{IEEEabrv,myrefs}

@article{perc2008social,
  title={Social diversity and promotion of cooperation in the spatial prisoner’s dilemma game},
  author={Perc, Matja{\v{z}} and Szolnoki, Attila},
  journal={Physical Review E},
  volume={77},
  number={1},
  pages={011904},
  year={2008},
  publisher={APS}
}

@article{xu2025reinforcement,
  title={Reinforcement learning can be a double-edged sword for cooperation on higher-order networks},
  author={Xu, Yan and Zhao, Dawei and Benko, Tina Perc and Xia, Chengyi and Perc, Matja{\v{z}}},
  journal={IEEE Transactions on Systems, Man, and Cybernetics: Systems},
  volume={56},
  number={1},
  pages={231--244},
  year={2025},
  publisher={IEEE}
}

@article{perc2010coevolutionary,
  title={Coevolutionary games—a mini review},
  author={Perc, Matja{\v{z}} and Szolnoki, Attila},
  journal={BioSystems},
  volume={99},
  number={2},
  pages={109--125},
  year={2010},
  publisher={Elsevier}
}

@article{jia2025asymmetric,
  title={Asymmetric interaction preference induces cooperation in human-agent hybrid game},
  author={Jia, Danyang and Dai, Xiangfeng and Xing, Junliang and Tao, Pin and Shi, Yuanchun and Wang, Zhen},
  journal={Science China Information Sciences},
  volume={68},
  number={11},
  pages={212201},
  year={2025},
  publisher={Springer}
}

@article{GDBJ1,
  author  = {Garrett Hardin},
  title   = {The Tragedy of the Commons},
  journal = {Science},
  volume  = {162},
  number  = {3859},
  pages   = {1243--1248},
  year    = {1968},
  doi     = {10.1126/science.162.3859.1243}
}

@article{jia2025social,
  title={Social networking agency and prosociality are inextricably linked in economic games},
  author={Jia, Danyang and Romi{\'c}, Ivan and Shi, Lei and Su, Qi and Liu, Chen and Liu, Jinzhuo and Holme, Petter and Li, Xuelong and Wang, Zhen},
  journal={Nature Human Behaviour},
  volume={9},
  number={12},
  pages={2620--2631},
  year={2025},
  publisher={Nature Publishing Group UK London}
}

@article{han2024co,
  title={Co-evolutionary dynamics in optimal multi-agent game with environment feedback},
  author={Han, Weiwei and Zhang, Zhipeng and Zhu, Yuying and Xia, Chengyi},
  journal={Neurocomputing},
  volume={581},
  pages={127510},
  year={2024},
  publisher={Elsevier}
}

@article{gong2022limit,
  title={Limit cycles analysis and control of evolutionary game dynamics with environmental feedback},
  author={Gong, Lulu and Yao, Weijia and Gao, Jian and Cao, Ming},
  journal={Automatica},
  volume={145},
  pages={110536},
  year={2022},
  publisher={Elsevier}
}

@article{liu2023coevolutionary,
  title={Coevolutionary dynamics via adaptive feedback in collective-risk social dilemma game},
  author={Liu, Linjie and Chen, Xiaojie and Szolnoki, Attila},
  journal={Elife},
  volume={12},
  pages={e82954},
  year={2023},
  publisher={eLife Sciences Publications, Ltd}
}

@article{GDBJBY,
  author  = {Beardsley, Timothy M.},
  title   = {A Lifeboat for a Lake},
  journal = {BioScience},
  volume  = {53},
  number  = {8},
  pages   = {691--691},
  year    = {2003},
  month   = {08},
  doi     = {10.1641/0006-3568(2003)053[0691:ALFAL]2.0.CO;2}
}

@article{han2022institutional,
  title={Institutional incentives for the evolution of committed cooperation: ensuring participation is as important as enhancing compliance},
  author={Han, The Anh},
  journal={Journal of The Royal Society Interface},
  volume={19},
  number={188},
  pages={20220036},
  year={2022},
  publisher={The Royal Society}
}

@article{zhu2022nash,
  title={Nash equilibrium in iterated multiplayer games under asynchronous best-response dynamics},
  author={Zhu, Yuying and Xia, Chengyi and Chen, Zengqiang},
  journal={IEEE Transactions on Automatic Control},
  volume={68},
  number={9},
  pages={5798--5805},
  year={2022},
  publisher={IEEE}
}

@article{song2025emergence,
  title={Emergence of cooperation and commitment in optional prisoner’s dilemma},
  author={Song, Zhao and Han, The Anh},
  journal={Applied Mathematical Modelling},
  pages={116603},
  year={2025},
  publisher={Elsevier}
}

@inproceedings{wang2022modelling,
  title={Modelling the dynamics of regret minimization in large agent populations: a master equation approach.},
  author={Wang, Zhen and Mu, Chunjiang and Hu, Shuyue and Chu, Chen and Li, Xuelong},
  booktitle={Ijcai},
  volume={22},
  pages={534--540},
  year={2022}
}

@article{ramazi2020global,
  title={Global convergence for replicator dynamics of repeated snowdrift games},
  author={Ramazi, Pouria and Cao, Ming},
  journal={IEEE Transactions on Automatic Control},
  volume={66},
  number={1},
  pages={291--298},
  year={2020},
  publisher={IEEE}
}

@article{axelrod1980effective,
  title={Effective choice in the prisoner's dilemma},
  author={Axelrod, Robert},
  journal={Journal of Conflict Resolution},
  volume={24},
  number={1},
  pages={3--25},
  year={1980},
  publisher={Sage Publications Sage CA: Los Angeles, CA}
}

@article{hardin1998extensions,
  title={Extensions of ``the tragedy of the commons"},
  author={Hardin, Garrett},
  journal={Science},
  volume={280},
  number={5364},
  pages={682--683},
  year={1998},
  publisher={American Association for the Advancement of Science}
}

@article{sun2023state,
  title={State-dependent optimal incentive allocation protocols for cooperation in public goods games on regular networks},
  author={Sun, Zaiben and Chen, Xiaojie and Szolnoki, Attila},
  journal={IEEE Transactions on Network Science and Engineering},
  volume={10},
  number={6},
  pages={3975--3988},
  year={2023},
  publisher={IEEE}
}

@article{zhu2026evolutionary,
  title={Evolutionary Game Dynamics in Tripartite Crowdsourcing with Incentive Regulation},
  author={Zhu, Zishuai and Hua, Shijia and Liu, Linjie and Chen, Xiaojie},
  journal={IEEE Transactions on Knowledge and Data Engineering},
  year={2026},
  publisher={IEEE}
}

@article{hua2024coevolutionary,
  title={Coevolutionary dynamics of collective cooperation and dilemma strength in a collective-risk game},
  author={Hua, Shijia and Xu, Mingquan and Liu, Linjie and Chen, Xiaojie},
  journal={Physical Review Research},
  volume={6},
  number={2},
  pages={023313},
  year={2024},
  publisher={APS}
}

@article{lu2024hybrid,
  title={Hybrid reward-punishment in feedback-evolving game for common resource governance},
  author={Lu, Zhengyuan and Hua, Shijia and Wang, Lichen and Liu, Linjie},
  journal={Physical Review E},
  volume={110},
  number={3},
  pages={034301},
  year={2024},
  publisher={APS}
}

@article{Weiz,
  author  = {Joshua S. Weitz and Ceyhun Eksin and Keith Paarporn and Sam P. Brown and William C. Ratcliff},
  title   = {An oscillating tragedy of the commons in replicator dynamics with game-environment feedback},
  journal = {Proc. Natl. Acad. Sci. USA},
  volume  = {113},
  number  = {47},
  pages   = {E7518--E7525},
  year    = {2016},
  doi     = {10.1073/pnas.1604096113}
}

@article{zhu2025finite,
  title={Finite Strategy Switches of Coordinating and Anti-Coordinating Games on Weighted Networks},
  author={Zhu, Yuying and Zhang, Zhipeng and Xia, Chengyi and Li, Xiang and Chen, Zengqiang},
  journal={IEEE Transactions on Cybernetics},
  year={2025},
  publisher={IEEE}
}

@article{feng2023evolutionary,
  title={An evolutionary game with the game transitions based on the Markov process},
  author={Feng, Minyu and Pi, Bin and Deng, Liang-Jian and Kurths, J{\"u}rgen},
  journal={IEEE Transactions on Systems, Man, and Cybernetics: Systems},
  volume={54},
  number={1},
  pages={609--621},
  year={2023},
  publisher={IEEE}
}

@article{pi2025dynamic,
  title={Dynamic evolution of complex networks: A reinforcement learning approach applying evolutionary games to community structure},
  author={Pi, Bin and Deng, Liang-Jian and Feng, Minyu and Perc, Matja{\v{z}} and Kurths, J{\"u}rgen},
  journal={IEEE Transactions on Pattern Analysis and Machine Intelligence},
  year={2025},
  publisher={IEEE}
}

@article{zhu2023equilibrium,
  title={Equilibrium analysis and incentive-based control of the anticoordinating networked game dynamics},
  author={Zhu, Yuying and Zhang, Zhipeng and Xia, Chengyi and Chen, Zengqiang},
  journal={Automatica},
  volume={147},
  pages={110707},
  year={2023},
  publisher={Elsevier}
}

@article{Tilman2018EvolutionaryGW,
  author  = {Andrew R. Tilman and Joshua B. Plotkin and Erol Akçay},
  title   = {Evolutionary games with environmental feedbacks},
  journal = {Nat. Commun.},
  volume  = {11},
  pages   = {915},
  year    = {2020},
  doi     = {10.1038/s41467-020-14731-4}
}

@article{10.1093/pnasnexus/pgae455,
  author  = {Hiromu Ito and Masato Yamamichi},
  title   = {A complete classification of evolutionary games with environmental feedback},
  journal = {PNAS Nexus},
  volume  = {3},
  number  = {11},
  pages   = {pgae455},
  year    = {2024},
  doi     = {10.1093/pnasnexus/pgae455}
}

@article{hua2023facilitating,
  title={Facilitating the evolution of cooperation through altruistic punishment with adaptive feedback},
  author={Hua, Shijia and Liu, Linjie},
  journal={Chaos, Solitons \& Fractals},
  volume={173},
  pages={113669},
  year={2023},
  publisher={Elsevier}
}

@article{GEZHONGCUJINFANGSHI,
  author  = {Martin A. Nowak},
  title   = {Five Rules for the Evolution of Cooperation},
  journal = {Science},
  volume  = {314},
  number  = {5805},
  pages   = {1560--1563},
  year    = {2006},
  doi     = {10.1126/science.1133755}
}

@article{ito2018scaling,
  title={Scaling the phase-planes of social dilemma strengths shows game-class changes in the five rules governing the evolution of cooperation},
  author={Ito, Hiromu and Tanimoto, Jun},
  journal={Royal Society open science},
  volume={5},
  number={10},
  year={2018},
  publisher={The Royal Society}
}

@book{tanimoto2015fundamentals,
  title={Fundamentals of evolutionary game theory and its applications},
  author={Tanimoto, Jun},
  year={2015},
  publisher={Springer}
}

@article{wang2015universal,
  title={Universal scaling for the dilemma strength in evolutionary games},
  author={Wang, Zhen and Kokubo, Satoshi and Jusup, Marko and Tanimoto, Jun},
  journal={Physics of Life Reviews},
  volume={14},
  pages={1--30},
  year={2015},
  publisher={Elsevier}
}

@article{JILI1,
  author  = {Wenqiang Zhu and Qiuhui Pan and Sha Song and Mingfeng He},
  title   = {Effects of exposure-based reward and punishment on the evolution of cooperation in prisoner's dilemma game},
  journal = {Chaos Solitons Fractals},
  volume  = {172},
  pages   = {113519},
  year    = {2023},
  doi     = {10.1016/j.chaos.2023.113519}
}

@article{JILI3,
  author  = {Xiaojie Chen and Attila Szolnoki},
  title   = {Punishment and inspection for governing the commons in a feedback-evolving game},
  journal = {PLOS Comput. Biol.},
  volume  = {14},
  number  = {7},
  pages   = {e1006347},
  year    = {2018},
  doi     = {10.1371/journal.pcbi.1006347}
}

@article{CHENGFAGUODU,
  author  = {Subhasish M. Chowdhury and Matteo M. Marini},
  title   = {Overbidding and heterogeneous behavior in contest experiments: A meta-comment on cross-cultural differences},
  journal = {J. Behav. Exp. Econ.},
  volume  = {110},
  pages   = {102210},
  year    = {2024},
  doi     = {10.1016/j.socec.2024.102210}
}

@article{FUZHIFANGCHENG,
  author  = {Peter D. Taylor and Leo B. Jonker},
  title   = {Evolutionary stable strategies and game dynamics},
  journal = {Math. Biosci.},
  volume  = {40},
  number  = {1-2},
  pages   = {145--156},
  year    = {1978},
  doi     = {10.1016/0025-5564(78)90077-9}
}

@article{JILIYINRU,
  author  = {Karl Sigmund and Christoph Hauert and Martin A. Nowak},
  title   = {Reward and punishment},
  journal = {Proc. Natl. Acad. Sci. USA},
  volume  = {98},
  number  = {19},
  pages   = {10757--10762},
  year    = {2001},
  doi     = {10.1073/pnas.161155698}
}

@article{CHENGFA2,
  author  = {Christoph Hauert and Arne Traulsen and Hannelore Brandt and Martin A. Nowak and Karl Sigmund},
  title   = {Via Freedom to Coercion: The Emergence of Costly Punishment},
  journal = {Science},
  volume  = {316},
  number  = {5833},
  pages   = {1905--1907},
  year    = {2007},
  doi     = {10.1126/science.1141588}
}

@article{JILI4,
  author  = {Arindam Mandal and Sukanta Sarkar and Sagar Chakraborty and Partha Sharathi Dutta},
  title   = {Mitigating ecological tipping points via game–environment feedback},
  journal = {Proc. R. Soc. A},
  volume  = {481},
  number  = {2312},
  pages   = {20240915},
  year    = {2025},
  doi     = {10.1098/rspa.2024.0915}
}

@article{NYCZKA2012317,
  author  = {Piotr Nyczka and Jerzy Cisło and Katarzyna Sznajd-Weron},
  title   = {Opinion dynamics as a movement in a bistable potential},
  journal = {Physica A},
  volume  = {391},
  number  = {1},
  pages   = {317--327},
  year    = {2012},
  doi     = {10.1016/j.physa.2011.07.050}
}

@article{chen2014probabilistic,
  title={Probabilistic sharing solves the problem of costly punishment},
  author={Chen, Xiaojie and Szolnoki, Attila and Perc, Matja{\v{z}}},
  journal={New Journal of Physics},
  volume={16},
  number={8},
  pages={083016},
  year={2014},
  publisher={IOP Publishing}
}

@article{wang2025optimally,
  title={Optimally combined incentive for cooperation among interacting agents in population games},
  author={Wang, Shengxian and Cao, Ming and Chen, Xiaojie},
  journal={IEEE Transactions on Automatic Control},
  volume={70},
  number={7},
  pages={4562--4577},
  year={2025},
  publisher={IEEE}
}

@article{zheng2021modeling,
  title={Modeling and dynamics of networked evolutionary game with switched time delay},
  author={Zheng, Yating and Li, Changxi and Feng, Jun-e},
  journal={IEEE Transactions on Control of Network Systems},
  volume={8},
  number={4},
  pages={1778--1787},
  year={2021},
  publisher={IEEE}
}

@article{zhao2016matrix,
  title={A matrix approach to the modeling and analysis of networked evolutionary games with time delays},
  author={Zhao, Guodong and Wang, Yuzhen and Li, Haitao},
  journal={IEEE/CAA Journal of Automatica Sinica},
  volume={5},
  number={4},
  pages={818--826},
  year={2016},
  publisher={IEEE}
}

@article{hu2023evolutionary,
  title={Evolutionary games with two species and delayed reciprocity},
  author={Hu, Kaipeng and Li, Zhouhong and Shi, Lei and Perc, Matja{\v{z}}},
  journal={Nonlinear Dynamics},
  volume={111},
  number={8},
  pages={7899--7910},
  year={2023},
  publisher={Springer}
}

@article{shijiagoverning,
  author={Hua, Shijia and Liu, Linjie},
  title={Governing the commons via adaptive incentive control with real-time feedback},
  journal={SCIENCE CHINA Information Sciences},
    year={2026},
  publisher={Science China Press}
}



\end{document}